\documentclass[preprint,12pt]{article}

\usepackage{amsmath}
\usepackage{amssymb} % note that w/LNCS there is always a warning with the \vec command, even if you remove usage in your doc
\usepackage{graphicx}
\usepackage{mathtools}
\usepackage{xspace}
\usepackage{enumitem} % Max: pulling bullets left
\usepackage[skins]{tcolorbox}
\usepackage{subcaption}
\usepackage{wrapfig}% 
\usepackage{url}
\usepackage{hyperref}
\usepackage{amsthm}

\newcommand{\numjobs}{\textit{num}\xspace}
\newcommand{\thresh}{\textrm{PRICE}\xspace}
\newcommand{\currentprice}{current price\xspace}
\newcommand{\price}{price\xspace}
\newcommand{\intvl}{\ensuremath{\mathcal{I}}}
\newcommand{\fee}{attached fee\xspace}

\newcommand{\defn}[1]{\textbf{\emph{#1}}}
\newcommand{\whp}{w.h.p.\xspace}
\newcommand{\floor}[1]{\left \lfloor #1 \right \rfloor}
\newcommand{\ceil}[1]{\left \lceil #1 \right \rceil}

\newcommand{\Linear}{\textsc{LINEAR}\xspace}
\newcommand{\LinPow}{\textsc{LINEAR-POWER}\xspace}

\newcommand{\algExp}{\alpha\xspace}

\newcommand{\estGap}{\gamma\xspace}

\newcommand{\AlgTotal}{A\xspace}
\newcommand{\AdvTotal}{B\xspace}

\newcommand{\AlgRB}{A_{\mbox{\tiny RB}}\xspace}

\newtheorem{theorem}{Theorem}
\newtheorem{fact}{Fact}
\newtheorem{lemma}{Lemma}

\newif\ifcomments
\commentstrue
\commentsfalse

\title{Bankrupting DoS Attackers\footnote{A preliminary version of this work appeared at the 32nd International Colloquium on Structural Information and Communication Complexity (SIROCCO) \cite{bankrupting:chakraborty}.}}

\author{
Trisha Chakraborty\thanks{Amazon Web Services, Minneapolis, MN, USA}
\and
Abir Islam\thanks{Meta Research, Menlo Park, CA 94025, USA}
\and
Valerie King\thanks{Department of Computer Science, University of Victoria, Victoria, BC V8P~5C2, Canada}
\and
Daniel Rayborn\thanks{Department of Computer Science and Engineering, Mississippi State University, Mississippi State, MS 39762, USA}
\and
Jared Saia\thanks{Department of Computer Science, University of New Mexico, Albuquerque, NM 87131, USA}
\and
Maxwell Young\footnotemark[4]
}

\date{}

\date{}
\begin{document}
\maketitle

\begin{abstract}
Can we make a denial-of-service attacker pay more than the server and honest clients? Consider a model where a server sees a stream of jobs sent by either honest clients or an adversary.  The server sets a price for servicing each job with the aid of an \emph{estimator}, which provides approximate statistical information about the distribution of previously occurring good jobs. 

We describe and analyze pricing algorithms for the server under different models of synchrony,  with total cost parameterized by the accuracy of the estimator. Given a reasonably accurate estimator, the algorithm's cost provably grows more slowly than the attacker's cost, as the attacker's cost grows large. Additionally, we prove a lower bound, showing that our pricing algorithm yields asymptotically tight results when the estimator is accurate within constant factors.
\end{abstract}

\clearpage
\section{Introduction} \label{s:intro}

Since their first documented occurrence in 1974~\cite{first-DoS-attack}, denial-of-service(DoS) attacks have evolved into a multi-billion dollar security dilemma~\cite{cost-DDoS}.  The DoS problem has always been economic: success or failure is decided by who can marshal resources most effectively.  In this paper, we provide a novel theoretical formulation and analysis of the DoS problem. We hope this will be a first step towards: (1) leveraging tools from theoretical computer science to solve a critical and persistent security problem; and (2) spurring new theoretical techniques and models; for example, in the same way that the theoretical focus on the paging problem eventually helped create the area of online algorithms.  

It has long been known that fees can be used to curb DoS attacks \cite{dwork:pricing,juels:dos,walfish:ddos,noureddinerevisiting}. 
This fee may be in the form of computation: a client solves a computational puzzle or Captcha per request \cite{kandula2005botz,dwork1992pricing,liu:proof}; or bandwidth: a client sends multiple messages to the server per request \cite{walfish2010ddos,walfish:ddos,yu2007detection}.   Fees may be set to $0$ in the absence of attack, so that in peaceful times, clients pay nothing.

%\cite{walfish2010ddos,walfish:ddos,yu2007detection,kandula2005botz}.

In this paper, we design algorithms for setting the fee (the ``price") in an online manner, assisted by an ``estimator", which provides statistical information about the distribution of previously occurring good jobs in the stream. Our goal is to set prices so as to minimize the total algorithm's cost - that is, the cost to the server and honest clients - relative to the attacker's cost.  To achieve this, we define an accuracy measure of the estimator, and determine how the algorithm's cost increases with this measure. 

\subsection{The Dynamic Job Pricing Problem}\label{s:problem}

\noindent{\bf Server and Clients.} 
The \defn{server} receives a stream of jobs.  Each job is either \defn{good}: sent by a client - to avoid redundancy, henceforth ``clients" rather than ``honest clients" - or \defn{bad}: sent by the adversary (described below).

We assume a message-passing model of communication.  The following happens in parallel for the server and clients. 
\begin{itemize}
    \item \textbf{Server:} The server uses the estimator to compute the \currentprice. Jobs are serviced based on whether the attached fee meets this \price. When the server receives a job request with an insufficient fee, it sends the \currentprice back to the sender.
    \item \textbf{Clients:} Each client may send a job to the server with a fee.
\end{itemize}

\medskip

\noindent{\bf Adversary.} The adversary knows everything about: the distribution of good jobs, the estimator, and our algorithm. At any time, the adversary can send messages to the server with (bad) job requests with fees attached.
\medskip

\noindent{\bf Costs.} Our algorithmic costs are two-fold: (1) fees paid by clients; and (2) \defn{service costs} incurred by the server.  We assume a normalized service cost of $1$ for servicing any job, whether good or bad.  The cost to the adversary is the total cost of all fees paid to the server.\medskip

%Ideally, one would prefer to assume an estimator which is as weak as possible.  Here we describe a set of \emph{properties} which suffice to prove our results.

\noindent{\bf Estimator.} 
The server has an estimator whose output may be used to set prices.   The \defn{estimator} estimates the number of good jobs generated in any contiguous time interval in the past.  For any such interval, {\boldmath{$\intvl$}}, let {\boldmath{$g(\intvl)$}} be the number of good jobs in $\intvl$ and {\boldmath{$\hat{g}(\intvl)$}} be the value returned by the estimator for interval $\intvl$. The estimator $\hat{g}$ is additive: for any two non-overlapping intervals, $\intvl_1$ and $\intvl_2$:
$$\hat{g}(\intvl_1 \cup \intvl_2) = \hat{g}(\intvl_1) + \hat{g}(\intvl_2).$$ 

For a given input, the \defn{estimation gap},   {\boldmath{$\estGap$}}, is the smallest number greater than or equal to $1$, such that the following holds for all intervals, $\intvl$ in the input:
$$g(\intvl)/\estGap - \estGap \leq \hat{g}(\intvl) \leq \estGap g(\intvl) + \estGap.$$

\medskip

\noindent{\bf Communication Latency.} We have two problem variants:
\begin{itemize}
    \item \textbf{Zero-latency:} All messages are sent instantaneously.  So, the clients always know the exact \currentprice; and the fees attached to jobs sent to the server always arrive before any change in the \currentprice, and thus are always serviced.\medskip

    \item \textbf{Asynchronous:} Communication is \emph{asynchronous}~\cite{lynch1996distributed}: there is some unknown maximum latency, {\boldmath{$\Delta$}} seconds, for any message to be received after it is sent. This latency can also capture the maximum time it took any client to attach a fee.  Also, there is some unknown maximum number of good jobs, {\boldmath{$M$}}, that can be generated over $\Delta$ seconds.  Both $\Delta$ and $M$ are known to the adversary, but not to the algorithm. Also, the adversary controls the timing of all messages subject to the above constraints. 
    \medskip
\end{itemize}

\noindent{\bf (Lack of) Knowledge of the Servers and Clients.} Other than the outputs of the estimator, the server and clients have no prior knowledge of the number of good jobs, the amount spent by the adversary, or the estimation gap $\estGap$. Recall that, in the latency model, $\Delta$ and $M$ are also unknown. 

\subsection{Model Discussion}\label{s:model-discussion}

\noindent\textbf{Estimator.} Our estimator generalizes many DoS detection tools, summarized below; see also ~\cite{khalaf2019comprehensive,li2023comprehensive} for comprehensive surveys.  DoS detection tools can be characterized along $3$ main axes.  First, input attributes used include: network traffic patterns~\cite{doriguzzi2020lucid,yuan2017deepdefense,wang2008intelligent,nguyen2010proactive,tsai2010early,braga2010lightweight,cheng2012extreme,yan2018multi,tan2013system,cheng2002use,li2009new,dainotti2009cascade}; reliability of access routers~\cite{gonzalez2011trust}; request payloads~\cite{ndibwile2015web,kim2006packetscore,yu2007detection}; geographical information~\cite{bandara2016preventing}; hop count travelled by request messages~\cite{chouhan2012hierarchical,wang2007defense}; and social media posts~\cite{chambers2018detecting}.  Second, detection algorithms used include: Bayesian learning~\cite{vijayasarathy2011system,katkar2015detection,gonzalez2011trust}; k-nearest neighbors~\cite{nguyen2010proactive,barrionuevo2018anomaly,li2008lightweight,bandara2016preventing,yu2007detection}; neural networks and deep learning~\cite{doriguzzi2020lucid,yuan2017deepdefense,chambers2018detecting,cheng2012extreme,yan2018multi}; and statistical methods~\cite{ndibwile2015web,tan2013system,cheng2002use,li2009new,dainotti2009cascade,wang2008intelligent}.  Finally, outputs of the detection algorithms include: classification of each packet as good or bad \cite{doriguzzi2020lucid,yuan2017deepdefense,vijayasarathy2011system,yan2018multi,tan2013system,cheng2002use,yu2007detection}; or general detection of the occurrence of an attack or the severity of attack \cite{katkar2015detection,wang2008intelligent,bandara2016preventing,nguyen2010proactive,tsai2010early,braga2010lightweight,chambers2018detecting,cheng2012extreme,yan2018multi,tan2013system,cheng2002use,li2009new,dainotti2009cascade}. 

The estimator generalizes three detection techniques. First, for detection tools that \textbf{classify jobs}: the estimator can return the number of jobs classified as good in the input interval, perhaps adjusted based on the classifier error rates if these are known.  Second, for detection tools that \textbf{detect attacks}: when there is no attack detected in the input interval, the estimator can return the number of jobs in the interval; when there is an attack detected, the estimator can return the expected number of good jobs based on past data before the attack.  Finally, when good jobs arrive via a \textbf{Poisson, Weibull, Gamma, or similar stochastic process}: we describe an estimator, in Section~\ref{s:estimator}, with estimation gap that is $O(\log g)$, with probability at least $1-1/g^k$ for any positive $k$, where $g$ is the number of good jobs over the system lifetime.

\medskip
\noindent\textbf{Prices.}  Many DoS defenses incorporate prices for service.  One such tool is \textbf{bandwidth-based pricing}~\cite{li2023comprehensive,walfish:ddos,walfish2010ddos,yu2007detection,lukyanenko2010playing}.  Every message from a client incurs a resource cost to the client; a bandwidth price for a job is set by requiring the client to send a certain number of messages per job serviced.  In the simplest case, the server services each request with some tunable probability.  This sets an \emph{expected} bandwidth cost for each job serviced.  When this probability changes, it is shared with the clients.  Speak-up~\cite{walfish:ddos,walfish2010ddos} (402 cites, according to Google Scholar) and DOW~\cite{yu2007detection} (112 cites) are popular DoS defenses that primarily use this approach.  When these systems detect an attack, they notify clients that the service probability has decreased.  If the attack continues, good and bad clients pay similar bandwidth costs, and also have similar rates of service; thus the total algorithmic cost is $\Omega(B)$.

Rate-limiting is a bandwidth-based defense technique where packets can be dropped with some tunable probability, but this probability is not necessarily shared with the clients~\cite{liu2021jaqen,zargar2013dicodefense,kim2006packetscore,kim2004packetscore}; OpenFlow~\cite{naous2008implementing,braga2010lightweight,yuhunag2010novel,buragohain2016flowtrapp,yao2011source} ($>1800$ cites) is a popular such rate-limiting system, even though it has no theoretical guarantees.

\medskip
\noindent\textbf{Puzzle-based pricing} defenses require clients to solve computational or proof-of-work (PoW) puzzles in order to receive service~\cite{dwork:pricing,juels:dos,aura2000resistant,Chen2010,khor2009spow,kumar2011mitigation,shi2006overdose,dixon2008phalanx,parno2007portcullis}.  Popular DoS defense systems using PoW include:  Poseidon~\cite{zhang2020poseidon} (207 cites); KaPow~\cite{feng2010kapow,kaiser2007mod_kapow} (37 cites); Phalanx~\cite{dixon2008phalanx} (120 cites); Portcullis ~\cite{parno2007portcullis} (288 cites); 
OverDOSe~\cite{shi2006overdose} (52 cites); and Puzzle Auctions~\cite{wang2003defending} (304 cites).  Other puzzle-based defenses use \emph{Captchas}~\cite{von2003captcha,von2008recaptcha}: human-solvable puzzles.  Such systems including: Kill-Bots~\cite{kandula2005botz} (572 cites) and Enhanced DDoS-MS~\cite{alosaimi2015denial} (13 cites).  In all of these puzzle-based systems, clients pay an amount during an attack that matches the attacker, and thus receive a fair share of job service.  However, to the best of our knowledge, none of these defense systems have provable total cost of $o(B)$, where $B$ is the amount spent by the attacker.

%The functionality of puzzle generation and solution verification can be hardened against attack by outsourcing to a dedicated service~\cite{waters:new} ($262$ cites); leveraging the Domain Name System \cite{Lee:DDD}; or using routers \cite{tourani2020persia}. 

%may be attacked~\cite{anderson2004preventing,wu-chang:case}, and solutions include outsourcing to a dedicated service \cite{waters:new} ($262$ cites), leveraging the Domain Name System \cite{Lee:DDD} or routers \cite{tourani2020persia}. 

\medskip
\noindent\textbf{Asymptotic Advantage.}
Many DoS defenses \textbf{filter} out packets that are classified as bad~\cite{6489876}.  But, even if filtering is $99.9\%$ effective -- and reported classification accuracy now hovers around $95\%$~\cite{doriguzzi2020lucid,ndibwile2015web,yuan2017deepdefense} -- filtering alone still requires the defenders to spend $\Omega(B)$.  Thus, standalone filtering will never give the defender an asymptotic advantage, although filtering can be combined with our results to reduce costs by constant factors.

Should DoS defenders care about an asymptotic advantage?  Resource costs related to DoS attacks are increasingly liquid: attackers can rent bot-farms on the black market~\cite{somani2017combating}; and defenders can provision additional resources from the cloud to service requests~\cite{yuan2020minimizing,zheng2019dynashield,AWS_white}.  So, if a defender has access to a DoS defense algorithm with an asymptotic cost advantage, it seems feasible they should be able to outlast an attacker with a similar monetary budget.  See also Section~\ref{s:related} for game-theoretic implications.

\subsection{Our Results}

For a given input, let {\boldmath{$g$}} be the number of good jobs,  {\boldmath{$\AdvTotal$}} be the total spent by the adversary, and $\estGap$ be the estimation gap for the input.  Our primary goal is to minimize the algorithmic cost --- i.e. fees plus  service costs --- as a function of $g$, $\AdvTotal$, and $\estGap$ 

Our first result is for the algorithm, \Linear, which assumes zero-latency.  %Our first result is the following (See Sections~\ref{s:linear} and~\ref{s:lin-anal}).

\begin{theorem}\label{t:Linear}
\Linear has total cost $$O\left( \estGap^{5/2}\sqrt{\AdvTotal\,(g+1)} + \estGap^3(g+1)\right).$$
Moreover, the server sends $O(g)$ messages to clients, and clients send a total of $O(g)$ messages to the server.  Also, each good job is serviced after at most $3$ messages are exchanged.
\end{theorem}

To build intuition for this result, consider the case where $\gamma = O(1)$. Then, the theorem states that the cost to \Linear is $O(\sqrt{\AdvTotal(g+1)} + (g+1))$.  Importantly, the cost to the algorithm grows asymptotically slower than $B$ whenever $g = o(B)$.  We note that the additive $O(g+1)$ term results from service costs due to good jobs, which must hold for any algorithm.  Our result is actually better than stated in Theorem~\ref{t:Linear}, when $\gamma$ grows large; in that case the total cost is always bounded by $g^2 + B$,  as will become clear in Section~\ref{s:algorithms}.

Finally, the number of messages sent from server to clients and from clients to server, and also the maximum number of messages exchanged to obtain service are all asymptotically optimal.  To see this note that at least one message must be sent for each good job serviced. 

\medskip

What if there is latency?  For our asynchronous variant, we design and analyze an algorithm \LinPow.  Our main theorem for this case follows.  The $\tilde{O}$ notation for the algorithmic cost hides terms that are logarithmic in $\AdvTotal$ and $\gamma$. 
 
\begin{theorem}\label{t:LinPow}
\LinPow has total cost:
$$ 
\tilde{O} \left(\estGap^3 \left(\sqrt{B+1} \min \left(g+1, M \sqrt{g+1},M \sqrt{B+1} \right) + (g+1) \right) \right).
$$
Moreover, the server sends $O( g \log (B + \estGap))$ messages to clients, and clients sends a total of $O( g \log (B + \estGap))$ messages to the server.  Also, each good job is serviced after at most $4\log (\estGap (B+1)) + 3$ messages are sent.  
\end{theorem}

To gain intuition for this result, assume that $\gamma$ and $M$ are both $O(1)$.  Then the cost to \LinPow is $\tilde{O} (\sqrt{B+1}$ $\min (\sqrt{g+1},$ $\sqrt{B+1}) + (g+1) )$.  This equals $\tilde{O} (\sqrt{B(g+1)} + (g+1))$,  since when $B = O(g)$, the expression is $\tilde{O}(g+1)$; and otherwise the expression is $\tilde{O}(\sqrt{B(g+1)})$.  Thus, when $\gamma$ and $M$ are constants, the cost of \LinPow matches that of \Linear, up to logarithmic terms.  For communication, the number of messages sent from the server to clients, clients to the server, and the number of messages exchanged until a job is serviced all increase by a logarithmic factor in $B$, when compared to \Linear.

\medskip
\noindent To complement these upper bounds, in Section \ref{s:lower}, we also prove the following lower bound.  This lower bound holds for our zero-latency variant, and so directly also holds for the harder, asynchronous variant. Let {\boldmath{$n$}} be the total number of jobs. % and {\boldmath{$\AlgTotal$}}  be the total algorithm cost.

\begin{theorem}\label{t:main-lower}
Let $\estGap$ be any positive integer; $g_0$ be any positive multiple of $\estGap$; and $n$ any multiple of $\estGap g_0$. Next, fix any randomized algorithm.  Then, there is an input  with $n$ jobs, $g$ of which are good for $g \in \{ \estGap g_0, g_0/\estGap \}$; an estimator with gap $\estGap$ on that input; and the randomized algorithm on that input has expected cost: $$ \Omega\left(\sqrt{\estGap B (g+1)} + \estGap (g+1)\right).$$
\end{theorem}

When $\gamma = O(1)$, this lower bound asymptotically matches the cost for \Linear, and is within logarithmic factors of the cost for \LinPow.  Additionally, the lower bound allows for a fair amount of independence between the settings of $\gamma$, $g$ and $n$ for which the lower bound can be proven.  For example, it holds both when $g$ is small compared to $n$ and also when $g$ is large.  Finally, note that the lower bound holds for randomized algorithms, even though both \Linear and \LinPow are deterministic.  Thus, randomization does not offer significant improvements for our problem.

Finally,  in Section \ref{s:simulation}, we provide preliminary simulation results for \Linear and \LinPow. Our results suggest that the hidden constants in the asymptotic notation may be small, and that the scaling behavior for our algorithms' costs matches that predicted by theory.

%%%%%%%%%%%%%%%%%%%%%%%%%%%%%%%%%%%%%
%%%%%%%%%%%%%%%%%%%%%%%%%%%%%%%%%%%%%
%%%%%%%%%%%%%%%%%%%%%%%%%%%%%%%%%%%%%

\subsection{Technical Overview}\label{s:overview} 

Our full proofs are provided in Sections \ref{s:lin-anal}, \ref{s:linPow-anal}, and \ref{s:lower}. Here, we give intuition for our technical results.  As a warm up, we start with our easiest result: the zero-latency case.  Then, we sketch how we build on this result in our asynchronous setting.  Finally, we give intuition for our lower bound.
\medskip

\noindent
\textbf{Zero-Latency.}  Our algorithm for the zero-latency problem is \Linear.  Here, we give a high-level overview of the algorithm and its analysis; our detatiled pseudocode is given in  Section~\ref{s:linear}, and our formal analysis is provided  in Section \ref{s:lin-anal}.  

\Linear is conceptually simple.  Time is partitioned into \defn{iterations}, which end whenever the estimator returns a value of at least $1$ for the estimate of the number of good jobs seen in the iteration so far.  The \currentprice is set to $s$, where $s$ is the number of jobs seen so far in the iteration.

Our analysis of \Linear involves two key steps: (1) establishing upper and lower bounds on the number of iterations as a function of $\estGap$ and $g$  (Lemmas \ref{l:upperIter} and \ref{l:lowerIter}); and (2) deriving an upper  bound on the number of good jobs in any iteration (Lemma \ref{l:upperGood}).

Next, using (1) and (2), we prove an upper bound on the cost of \Linear.  To build intuition, we consider a general class of \currentprice setting algorithms where the charge to the $s$-th job in the iteration is $s^{\algExp}$, for any constant $\algExp \geq 0$.  We prove that $\algExp = 1$ gives the best algorithmic cost as a function of $B$, $g$ and $\estGap$.

Our upper bound consists of four exchange arguments~\cite{kleinberg2006algorithm}, which will be over the ordering of good and bad job requests.  Our first two exchange arguments lower bound $B$ and the next two upper bound the algorithm cost, {\boldmath{$A$}}.  We begin by lower bounding $B$.  To do so, our first exchange argument shows that $B$ is minimized when, in each iteration, the bad jobs occur before the good.  Our second exchange argument shows that, subject to this constraint, $B$ is minimized when the bad jobs are distributed as evenly as possible across the iterations.  From these two exchange arguments, a lower bound on $B$ follows via an integral lower bound and basic algebra (see Lemma \ref{l:costAdv}).

We next upper bound $A$.  To do so, we make two  observations: (1) within an iteration, $A$ is maximized by putting all good jobs at the end; and (2) without loss of generality, we can assume that the iterations are sorted in decreasing order by the number of bad jobs they contain.  Then, we use a third (simple) exchange argument to show that $A$ is maximized when good jobs are ``packed left'': packed as much as possible in the earlier iterations, which have the higher number of bad jobs.  
 
Finally, we use the fourth (harder) exchange argument to show that $A$ is maximized when {\it bad} jobs are ``packed left'' when $\algExp \geq 1$, and are evenly spread when $\algExp < 1$.  To show this, for an iteration $i$, we let $f(b_i)$ be the cost of the $i$-th iteration when there are $b_i$ bad jobs at the start.  Then, let $\delta(b_i) = f(b_i +1) - f(b_i)$ be the change in the cost of iteration $i$ when we add $1$ bad job.  Then the change in cost when there is an exchange of one bad job from iteration $k$ to iteration $i$ for $k>i$ will be $\delta(b_i) - \delta(b_k-1)$.

How do we show that $\delta(b_i) - \delta(b_k-1)) \geq 0$?  If we relax $b_i$ so it can take on values in the real numbers, then $f(b_i)$ becomes a continuous, differentiable function. Thus, we can use the mean value theorem (MVT). In particular, the MVT shows that $\delta(b_i)$ equals $f'(x)$ at some value of $x \in [b_i, b_{i+1}]$. Then, upper and lower bounding the $\delta$ values reduces to bounding $f'$ in the appropriate range.  This is a simpler problem because $f'$ is a monotonically non-decreasing (non-increasing) function when $\algExp \geq 1$ ($\algExp < 1$, respectively).  Once we have this exchange argument in hand, we can calculate the algorithmic costs (see the proofs of Lemma \ref{l:CostAlgRB} and Theorem \ref{t:Linear} in Section \ref{s:lin-anal} for details). 
\medskip

\noindent\textbf{Asynchronous.} If we naively use \Linear for the asynchronous model, costs can be high.  For simplicity, assume that $\estGap = \Theta(1)$ and $g = \Theta(M\sqrt{B})$.   Then, the adversary can insert $\sqrt{B}$ bad jobs initially, causing the \currentprice in \Linear to increase $\sqrt{B}$ times.  During this time $ \Theta\left(\min\left(g,M \sqrt{B} \right)\right)$ good jobs join the system, all of them are returned without service repeatedly, as the \currentprice increases, until the \currentprice reaches its maximum value of $\Theta(\sqrt{B})$.  Next, all $\Theta\left(\min(g, M \sqrt{B})\right)$ of these good jobs wind up paying a fee equal to $\Theta(\sqrt{B})$.  Adding in the cost to service jobs and the cost for any additional good jobs that enter after the bad, we have a total cost for the \Linear of $\Theta\left(\min \left(g+1, M \sqrt{B+1} \right) \sqrt{B+1} + g \right)$. 

How can we improve this?  One key idea is to reduce the number of times that the server increases the \currentprice.  Our second algorithm, \LinPow does just that: (1) the server sets the \currentprice to the largest power of $2$ less than the total number of jobs in the current iteration; and (2) clients double their fee whenever their job request is returned.  

To compare \LinPow with \Linear, suppose $\estGap = O(1)$.  Then, it is not hard to see that in any iteration $i$ where the adversary spends $B_i$, the maximum \currentprice  is $O(\sqrt{B_i})$ (see Lemma \ref{l:maxRBThresh}).  Thus, we can argue that any job is returned $O(\log B)$ times, and so pays a total of $O(\sqrt{B})$. 

But the above argument gives a poor bound on fees paid by the clients, since it multiplies all such fees from the \Linear analysis by a $\sqrt{B}$ factor.  To do better, we need a tighter analysis for fees.  But how do we proceed when a single job can be refused service over many iterations because it keeps learning the \currentprice too late?   

The first key idea is to partition time into epochs, for the purpose of analysis only.  An \defn{epoch} is a period of time ending with $2\Delta$ seconds during which the \currentprice never increases.  Then, several facts follow.  First, every good job is serviced within $2$ epochs.  Second, we can bound the total fees that all good jobs pay during epoch $i$ to be essentially $O(M (\sqrt{B_{i-1}} + \sqrt{B_{i}}))$, where $B_j$ is the amount spent by the adversary in epoch $j$  (see Lemma \ref{l:overpayEpoch}).  Third, an epoch has overpaying jobs only if it or its preceding epoch spans the beginning of an iteration (see Lemma \ref{l:nonZeroEpoch}).  This last fact allows us to bound the total number of epochs that add to the overpayment amount using our previous bounds on the number of iterations.  Putting all these facts together, we obtain a total cost for \LinPow that is only $\tilde{O} \big(\min \big(g+1, M \sqrt{g+1},M \sqrt{B+1} \big)\sqrt{B+1}$  $+ g \big)$  (see Section \ref{s:linPow-anal}) versus the cost of \Linear computed above of $\Theta(\min (g+1, M \sqrt{B+1} ) \sqrt{B+1} + g )$.  Ignoring logarithmic factors, the cost of \LinPow may be less than \Linear by a factor of $\sqrt{g+1}/M$.  As an example, when $M=O(1)$ and $B = \Theta(g)$, \LinPow has cost $\tilde{O}(g)$, and \Linear has cost $\Theta(g^{3/2})$.\medskip

\noindent\textbf{Lower Bound.} We prove a lower bound for the zero-latency case (Section \ref{s:lower}). This bound directly holds for our harder, asynchronous problem variant.

In our lower bound, for given values of $\estGap$, $g_0$ and $n$, we design an input distribution that (1) sets $g$ to  $g_0/\estGap$ with probability $1/2$, and sets $g$ to $\estGap g_0$ otherwise; (2) distributes $g$ good jobs uniformly but randomly throughout the input; and (3) creates an estimator that effectively hides both the value of $g$ and also the location of the good jobs, and always has estimation gap at most $\estGap$ on the input (see Section \ref{s:adv} and Lemma \ref{l:estim} for details).

The key technical challenge is to set up a matrix in order to use Yao's minimax principle~\cite{yao:probabilistic} to prove a lower bound on the expected cost for any randomized algorithm.  In particular, we need a game-theory payoff matrix where the rows are deterministic algorithms, and the columns are deterministic inputs.  A natural approach is for the matrix entries to be the ratios $A/\sqrt{B g}$, where $A$ is the cost to our algorithm.  Unfortunately, this fails to achieve our goal since this would yield $E[A/\sqrt{Bg}]$ and we need $E[A]/E[\sqrt{Bg}]$, which are not generally equal.  Some lower bound proofs can circumvent this type of problem.  For example, lower bound proofs on competitive ratios of online algorithms circumvent the problem because the offline cost does not vary with the choice of the deterministic algorithm, i.e. from row to row. So the expected offline cost can be factored out and computed separately (see~\cite{borodin2005online}, Chapter 5).  Unfortunately, this does not work for us since our random variables $A$ and $B$ are interdependent, and vary across both rows and columns.

To address this challenge, we first set the matrix entries to $E(A) - \sqrt{E(B) g \estGap})$, and then prove that the minimax for such a payoff matrix has expected value at least $0$.  Then, we use Jensen's inequality and linearity of expectation to show that the minimax is also at least $0$ when the matrix entries are $E(A - \sqrt{B g \estGap})$ (see Lemma \ref{l:lbDeterm}).  Next, using Yao's principle, we can show the maximin is nonnegative, and thereby obtain a key inequality for $A$ and $B$.   Once we have this inequality, we can use algebra to establish our lower bound for any randomized algorithm (see Theorem \ref{t:main-lower}).  

We believe that this new approach -- using Yao's principle on a matrix whose entries are $X - Y$ in order to bound the ratio $E(X)/E(Y)$, and thereby prove some target lower bound of $E(Y)$ -- may have broader applications to proving lower bounds for resource competitive algorithms (see our expanded discussion of related work in Section \ref{s:related}).

\subsection{Other Related Work} \label{s:related}

A preliminary version of this work appeared at the {\it 32nd International Colloquium on Structural Information and Communication Complexity (SIROCCO)} \cite{bankrupting:chakraborty}. Here, we present the full technical analysis, along with an in-depth description and intuition for our results, an expanded discussion of related work, and preliminary findings from our simulations.
In addition to the discussion given previously in Section \ref{s:model-discussion}, we summarize additional related results here.

\medskip

\noindent{\bf Game Theory.} Many results analyze security problems in a game theoretic setting.  These include two-player security games between an attacker and defender, generally motivated by the desire to protect infrastructure (airports, seaports, flights, networks)~\cite{tambe2011security,avenhaus2002inspection,von1991recursive} or natural resources (wildlife, crops)~\cite{fang2015security,fang2017paws}.  Game theoretic security problems may include multiple players, and may also have a mix of Byzantine, altruistic and rational (BAR) players~\cite{clement2007theory,clement2008bar}.  

Specific to the challenge of mitigating DoS attacks, there are several results that apply game theory (for examples, \cite{wu2010modeling,noureddinerevisiting,SPYRIDOPOULOS201339,fallah2008puzzle,dingankar2007denial,he:game}).   In contrast to these results, our approach does not assume the adversary, server or clients are rational.  Our main results show how to minimize the server and clients costs \emph{as a function of the adversarial cost}.  This is different than optimizing most typical utility functions for the defender.  However, there are two game theoretic connections.  First, when there exists an algorithm for the defender, as in this paper, where the defender pays some function $f$ of the adversarial cost $B$, and $f(B) = f(0) + o(B)$, it is possible to solve a natural zero-sum, two-player attacker-defender game. The solution of this game is for the attacker to set $B = 0$ --- essentially to not attack --- in which case the defender pays a cost of $f(0)$, which in our case is $O(\gamma^3 g)$  (see Section 2.6 of~\cite{gupta2020resource} for details). Second, our paper suggests, for future work, the game-theoretic problem of refining our current model and algorithms in order to build a mechanism for clients and a server that are all selfish-but-rational (Section~\ref{s:conclusion-future-work}).

\smallskip

\noindent\textbf{General Resource Burning.}  Starting with the seminal 1992 paper by Dwork and Naor~\cite{dwork:pricing}, there has been significant academic research, spanning multiple decades, on using resource burning (RB) -- the verifiable expenditure of resources -- to address many security problems; see surveys~\cite{alifoundations,gupta2020resource}.  Security algorithms that use resource burning arise in the domains of wireless networks~\cite{gilbert:sybilcast}, peer-to-peer systems~\cite{li:sybilcontrol,borisov:computational}, blockchains~\cite{lin2017survey}, and e-commerce~\cite{gupta2020resource}. Our algorithm is purposely agnostic about the specific resource burned, which can include computational power~\cite{wang:defending}, bandwidth~\cite{walfish2010ddos},  computer memory~\cite{abadi2005moderately}, or human effort~\cite{von2003captcha,oikonomou2009modeling}.  In the DoS setting, the functionality of puzzle generation and solution verification can be hardened against attack by outsourcing to a dedicated service~\cite{waters:new}; leveraging the Domain Name System \cite{Lee:DDD}; or using routers \cite{tourani2020persia}. \smallskip

\noindent{\bf Resource-Competitive Algorithms.} %\label{sec:rc-algs}
There is a growing body of research on security algorithms whose costs are parameterized by the adversary's cost.  Such results are  referred to as {\it resource-competitive}~\cite{Bender:2015:RA:2818936.2818949}.  Specific  results in this area address: communication on a broadcast channel \cite{gilbert:making,gilbert:near,king:conflict,chen:broadcasting}, contention resolution~\cite{bender:how,bender:fully}, interactive communication~\cite{ICALP15,aggarwal2016secure}, hash tables \cite{chakraborty:defending}, bridge assignment in Tor~\cite{zamani2017torbricks}, and the Sybil attack~\cite{Gupta_Saia_Young_2021,GUPTA202389,pow-without}. Our result is the first in this area to use an estimator and incorporate the estimation gap into the algorithmic cost.\smallskip

\noindent{\bf Algorithms with Predictions.}  Algorithms with predictions is a new research area which seeks to use predictions to achieve better algorithmic performance.  In general, the goal is to design an algorithm that performs well when prediction accuracy is high, and where performance drops off gently with decreased prediction accuracy. Critically, the algorithm has no a priori knowledge of the accuracy of the predictor on the input.

The use of predictions offers a promising approach for improving algorithmic performance.  This approach has found success in the context of contention resolution~\cite{gilbert:contention}; bloom filters~\cite{mitzenmacher:model}; online problems such as job scheduling~\cite{PurohitSK18,lattanzi2020online,DBLP:conf/innovations/ScullyGM22}; and many others.  See surveys by Mitzenmacher and Vassilvitskii~\cite{roughgarden2021,mitzenmacher:algorithms-acm}.  Our estimation gap for the input, $\estGap$ is related to a accuracy metric for a predictor of job lengths in~\cite{DBLP:conf/innovations/ScullyGM22}.  

Our estimator differs from predictions in that it never returns results predicting the future, it only estimates the number of good jobs in past intervals.

%%%%%%%%%%%%%%%%%%%%%%%%%%%%%%%%%%%%%%%%%%%%%%
%%%%%%%%%%%%%%%%%%%%%%%%%%%%%%%%%%%%%%%%%%%%%%
%%%%%%%%%%%%%%%%%%%%%%%%%%%%%%%%%%%%%%%%%%%%%%

\section{Our Algorithms}\label{s:algorithms}

In Section~\ref{s:linear}, we formally define the algorithm \Linear; and in Section~\ref{s:linPow}, we define the algorithm \LinPow.  As described in Section~\ref{s:problem}, both of our algorithms make use of an estimator. So, in Section~\ref{s:example-scenario}, we give an example of the type of estimator that might be used by these algorithms.  Next, in Section~\ref{s:scenarios}, to build intuition, we give some example scenarios for our algorithms. 
 
\subsection{Algorithm \Linear}\label{s:linear}
Our first algorithm, \Linear, addresses our zero-latency problem variant. The server partitions the servicing of jobs into \defn{iterations}, which end when $\hat{g}(\intvl) \geq 1$, where $\intvl$ is the time interval in the current iteration.  The \currentprice, \defn{\thresh},  for the next job is set to $s + 1$, where $s$ is the number of jobs serviced so far in the iteration. Since there is no latency, good jobs will always send in a fee with value equal to \thresh, thus jobs with lower value solutions can be ignored. 
 Pseudocode is given in Figure~\ref{alg:linear}.  
Again, because there is no latency, the client's action is just to either send in the job with attached fee equal to the current value of \thresh, or else drop the job. Every message is presumed to be received at a unique point in time; if not, the server can break ties arbitrarily.

%%%%%%%%%%%%%%%%%%%%%%%%%%%%%%%%%%%%%
%%%%%%%%%%%%%%%%%%%%%%%%%%%%%%%%%%%%%
\begin{figure}[t!]
\centering
\begin{tcolorbox}[standard jigsaw, opacityback=0]
\begin{minipage}[h]{0.97\textwidth}
\noindent{\textsc{\large \Linear }}
\medskip

\noindent{\bf Server}: \\
Repeat forever: 
\begin{enumerate}
\item If $\hat{g}(\intvl) \geq 1$, where $\intvl= (t,t']$,  $t$ is the start time of the current iteration, and $t'$ is the current time, then begin a new iteration.
\item $\thresh \leftarrow s+1$, where $s$ is the number of jobs serviced so far in this iteration. If $\thresh$ has changed, send the new value to all clients.
\item If the next job has a \fee at least equal to $\thresh$ then service it, otherwise do not service and reply with $\thresh$.  
\end{enumerate}
\smallskip
\noindent{\bf Client:} $x \leftarrow 1$.\\ Until the job is serviced, do:
 
\begin{enumerate}
\item Send the server a job request with fee $x$.
\item $x\leftarrow$ most-recent \price  received from the server.
\end{enumerate}

\end{minipage}
\end{tcolorbox}
\vspace{-10pt}\caption{Pseudocode for \Linear.}
\label{alg:linear}  
\end{figure}

%%%%%%%%%%%%%%%%%%%%%%%%%%%%%%%%%%%%%
%%%%%%%%%%%%%%%%%%%%%%%%%%%%%%%%%%%%%

\subsection{Algorithm \LinPow}\label{s:linPow}

We next describe a second algorithm, \LinPow, that addresses the asynchronous variant of our problem; see Figure~\ref{alg:linPow} for the pseudocode. 
\smallskip

\noindent
\textbf{\large Server.} The server operates in \defn{iterations}, which end when $\hat{g}(\intvl) \geq 1$, where $\intvl$ is the time interval elapsed in the iteration. 

In each iteration, the \currentprice, $\thresh$, is set to $2^{\lfloor\log (s+1)\rfloor}$, where $s$ is the number of jobs serviced so far in the iteration.  If the fee attached to the next job is at least $\thresh$, the job is serviced.  Otherwise, the job is not serviced, and the server sends a message with the value $\thresh$ to the client that sent the job. 

\smallskip
\noindent
\textbf{\large Client.}  
Initially, each client sets a local $x$ variable to $1$.  The client sends the server a fee of $x$ with its job.  When receiving a message from the server with a $thresh$ value greater than $x$, it sets $x$ to the value received, and re-sends the job with a fee of $x$, or chooses to drop the job.

%%%%%%%%%%%%%%%%%%%%%%%%%%%%%%%%%%%%%
%%%%%%%%%%%%%%%%%%%%%%%%%%%%%%%%%%%%%
\begin{figure}[t!]
\centering
\begin{tcolorbox}[standard jigsaw, opacityback=0]
\begin{minipage}[h]{0.97\textwidth}
\noindent{\textsc{\large \LinPow }}
\medskip

\noindent{\bf Server}:\\
Repeat forever:
\begin{enumerate}
\item If $\hat{g}(\intvl) \geq 1$, where $\intvl= (t,t']$,  $t$ is the start time of the current iteration, and $t'$ is the current time, then begin a new iteration.
\item $\thresh \leftarrow 2^{\floor{\lg (s+1)}}$, where $s$ is the number of jobs serviced in this iteration.
\item If the next job has a fee at least equal to \thresh, then service it. Otherwise, do not service the job; instead send a message with the value \thresh to the client that sent the job.
\end{enumerate}
\smallskip
\noindent{\bf Client:} $x \leftarrow 1$.\\ Until the job is serviced, do:
\begin{enumerate}
\item Send the server a fee of $x$ with the job.
\item $x\leftarrow$ maximum \price ever received from the server.
\end{enumerate}

\end{minipage}
\end{tcolorbox}
\vspace{-10pt}\caption{Pseudocode for \LinPow.} 
\label{alg:linPow} 
\end{figure}

\section{An Estimator with Broad Application}\label{s:example-scenario} \label{s:estimator}
 
Our main goal is to design algorithms that harness the power of estimators, rather than to create the estimators themselves. However, to convey the broad applicability of our results, we consider various distributions for the arrival of good jobs that correspond with real-world observations and established traffic models.  We describe an example estimator, with the properties defined in Section~\ref{s:problem}, which we show has small $\estGap$ for many types of distributions.  We then illustrate executions of both \Linear and \LinPow using this estimator.

\subsection{Poisson Estimator} \label{s:PoissonEstimator}

Assume good jobs are generated via a Poisson distribution with parameter $\lambda$, so that the expected interval length between good jobs is $1/\lambda$. Poisson job arrivals are one of the most commonly used statistical models, and are consistent with empirical findings~\cite{liu2001traffic,cao2003internet}.

For simplicity, we assume exact knowledge of the Poisson parameter $\lambda$, but a constant factor approximation suffices for the analysis below.  For an interval $\intvl$, let $\ell(\intvl)$ be the length of $\intvl$.  Then, for any interval $\intvl$, set
$$\hat{g}(\intvl) = \ell(\intvl)/ \lambda$$

\subsubsection{Bounding \texorpdfstring{{\boldmath{$\gamma$}}}{gamma}}\label{sec:bounding-gamma}
We now show that \whp in $g$, i.e. probability at least $1-1/g^k$ for any positive $k$, $\estGap = O(\log g)$ for the Poisson estimator.  This holds for \emph{all} distributions of bad jobs including: no bad jobs; more bad jobs than good; bad jobs uniformly distributed; and bad jobs that are clustered in any manner.  

As a special case, consider when the adversary inserts bad jobs using the same Poisson distribution as the good jobs. Then, it is impossible to differentiate a good job from a bad job with better than a (fair) coin toss.  Thus, the estimation gap remains $O(\log g)$ even though we can not classify jobs as good or bad. So, estimation can be easier than classification.\medskip

\noindent{\bf Analysis.} To formally bound $\gamma$, consider any input distribution with $g$ good jobs generated by the Poisson process.  We show that \whp in $g$, $\hat{g}$ ensures that $\estGap=O(\log g)$.

%We show that \whp in $g$, i.e. probability at least $1-1/g^k$ for any positive $k$, $\hat{g}$ ensures that $\estGap=O(\log g)$.
Recall that $g(\intvl)$ is the number of good jobs generated during $\intvl$.  We rely on the following fact which holds by properties of the Poisson process, Chernoff and union bounds.  In the following \whp means with probability of error that is polynomially small in $g$.

\begin{fact} \label{f:contigPoisson}
For some constant $C$, \whp for any interval $\intvl$: 
\begin{enumerate}
    \item If $g(\intvl) < C \ln g$, then $\ell(\intvl) = O((1/\lambda) \log g)$
    \item If $g(\intvl) \geq C \ln g$, then $\ell(\intvl) = \Theta((1/\lambda) g(\intvl))$.
\end{enumerate}
\end{fact}

We can now show the following simple fact.
\begin{fact} 
With high probability, the estimator $\hat{g}$ has estimation gap $\gamma = O(\log g)$.
\end{fact}
\begin{proof}
Consider two cases:\smallskip

\noindent {\it Case:} $g(\intvl) < C \ln g$. Then, by Fact~\ref{f:contigPoisson} (1), $\ell(\intvl) = O((1/\lambda) \log g)$, and so $\hat{g}(\intvl) = O(\log g)$.  Thus, there exists some $\estGap$, $\estGap = O(\log g)$ such that 
$$g(\intvl) - \estGap \leq \hat{g}(\intvl) \leq g(\intvl) + \estGap$$
Hence, the estimation gap on this type of interval is $O(\log g)$.\smallskip

\noindent
{\it Case:} $g(\intvl) \geq C \ln g$.  Then by Fact~\ref{f:contigPoisson} (2), \whp, $\ell(\intvl) = \Theta((1/\lambda) g(\intvl))$, and so $\hat{g}(\intvl) = \Theta(g(\intvl))$.  So, there exists some $\estGap$, $\estGap = O(1)$ such that 
$$(1/\estGap) g(\intvl) \leq \hat{g}(\intvl) \leq \estGap g(\intvl)$$
Hence, the estimation gap on this type of interval is $O(1)$. We conclude that for this example input distribution, \whp,  $\hat{g}$ is an estimator with $\estGap = O(\log g)$.
\end{proof}

\medskip

\noindent{\it Estimating  $\lambda$.}  A simple way to estimate $\lambda$ is to use job arrival data from a similar day and time.  To handle possible spikes in usage by clients, one may need to analyze data on current jobs, including: packet header information~\cite{carl2006denial}, geographic client location~\cite{feng2010kapow}, timing of requests \cite{ChakrabortyMM23}, and spikes in traffic rate~\cite{lakhina2004diagnosing}. 

\subsection{An Estimator for Weibull, Gamma, and Distributions with Exponentially Bounded Moments}

The $\gamma = O(\log g)$ result for the Poisson distribution can be generalized to arrival times given by other distributions whose moments do not grow super-exponentially.

In particular, for any $i \geq 1$, let $X_i$ be a random variable giving the time elapsed between the $i-1$ and $i$-th arrival.  Assume that the $X_i$ values are independent and identically distributed.  Then, let $Y_i = X_i - E(X_i)$, and assume that for some real positive $L$ and every positive integer $k$,
$$E(|Y^k_i|) \leq \frac{1}{2} E(Y^2_i)L^{k-2}k!$$

In other words, the above says that the $k$-th moment of $Y_i$ is bounded by an exponential in $k$.  This is true for common arrival time processes such as the Weibull~\cite{cao2004stochastic} and the gamma distributions~\cite{choi1999behavioral} with fixed parameters.

Then set $\rho = 1/E(X_i)$, and we can use an estimator that for any interval $\intvl$, sets
$$\hat{g}(\intvl) = \rho \ell(\intvl)$$

We can again show the following.
\begin{fact} 
With high probability, the estimator $\hat{g}$ has estimation gap $\gamma = O(\log g)$.
\end{fact}

To show this, we use Bernstein's inequality (specifically, the second inequality in \cite{bernstein-inequalities}; see also~\cite{uspensky1937introduction}) to prove a result analogous to Fact~\ref{f:contigPoisson}, but with $1/\lambda$ replaced by $\mu$. Then, the result follows directly.

\subsection{Adversarial Queuing Theory Estimator}

To capture a range of network settings, adversarial queuing theory (AQT) \cite{borodin2001adversarial} parameterizes the behavior of packet arrivals by an injection rate,  $\rho$ where $0<\rho \leq 1$, and a non-negative integer, {\boldmath{$b$}}, that denotes burstiness.  In AQT, a common arrival model is the \defn{leaky-bucket adversary}\cite{cruz-one} where, for any $t$ consecutive slots, at most $\rho t + b$ packets may arrive, for some number $b$ (for examples, see \cite{garncarek:stable,garncarek2018local,anantharamu2010deterministic,chlebus2009maximum,anantharamu2015broadcasting}).  The physical metaphor is that there is a bucket of size $b$ gallons.  Water enters the bucket according to the (adversarial) arrival process and it leaves at a constant rate of $\rho$.  The model requires that $b$ be large enough so that the bucket never overflows.

We define a new variant, the \defn{rain-barrel model}, that both lower bounds and upper bounds the number of arriving good jobs.  In particular, we require that: for any interval of $t$ seconds, between $(1/b) \rho t - b$ and $b \rho t + b$ good jobs arrive.  The physical metaphor is that there is a barrel of size $b$ and an additional reservoir of $b$ gallons that can be added to the barrel at any time.  Water enters the barrel according to the (adversarial) arrival process and leaves at a tunable rate in the range $\rho\cdot[1/b,b]$.  The new requirement is that there is always water available, but that the barrel never overflows.

In such a model, with knowledge of $\rho$, it is easy to design an estimator with estimation gap $\gamma = b$.  In particular, again set $\hat{g}(\mathcal{I}) = \rho \,\ell(\mathcal{I})$. By the definition of the bounded bucket model, this estimator has estimation gap $\gamma = b$.

Notably, the value $b$ need not be known a priori.  Further,  
as with the previous estimators, using a constant factor approximation to $\rho$ will not change the asymptotic value of $\gamma$.

\subsection{Zero Latency and Asynchronous Example Runs}\label{s:scenarios} 

We now illustrate the execution of our zero-latency and asynchronous algorithms under Poisson arrivals, using our first estimator (Section \ref{s:PoissonEstimator}); results are similar for our other two estimators.\smallskip

%We give example runs for our zero-latency and asynchronous algorithms  under Poisson arrivals.

\noindent{\bf Zero-Latency.} First, suppose the estimator is perfectly accurate.  Then, $\hat{g}(\intvl)$ is $0$ during an iteration until a good job arrives.  Once a good job arrives, $\hat{g}(\intvl)$ increases from $0$ to $1$, and a new iteration begins, starting with this new good job.  The value $\thresh$ is set back to $1$, and the good job pays a fee of $1$.  Then, the cost of each good job is $1$ and the cost of each bad job is the number of bad jobs received since the last good job was received. There are exactly $g$ iterations.  If $b$ bad jobs are submitted, the adversary's costs are minimized when the bad jobs are evenly spread over the iterations, so the cost for each bad job ranges from $1$ to $b/g$ in each of the $g$ iterations, for a total cost of $B=\Theta(b^2/g)$.

Now, consider the special case where the arrival of good jobs is a Poisson process with parameter $\lambda$, and we use the example estimator from Section~\ref{s:PoissonEstimator}.  Figure~\ref{f:estimator-example} illustrates this scenario. Here, iterations are demarcated by multiples of $1/\lambda$ and \thresh is reset to $1$ at the end of each iteration. Many bad jobs may be present over these iterations, and there may be more than $1$ good job per iteration (as depicted).

The lengths of time between good jobs are independently distributed, exponential random variables with parameter $\lambda$ (see \cite{mitzenmacher2017probability}, Chapter 8) that is, the probability that $j$ good jobs occur in time $t$ is $e^{(-t\lambda)} \frac{(t\lambda)^j}{j^!}$.
The adversary can decide where to place bad jobs. 
We now consider what happens with \Linear.   When using the example estimator from Section~\ref{s:PoissonEstimator}, for all $i\geq 1$, the $i$-th iteration ends at time $i/\lambda$, and the \currentprice is reset to $1$.
\begin{enumerate} 
\item   If there are no bad jobs, the expected cost of the $i$-th iteration to the algorithm is $O(X_i^2)$ where $X_i$ is the number of good jobs in a period of length $1/\lambda$, and $\sum_i X_i = g$. 
Since $E[X_i^2]= \sum_j  j^2 e^{(-\lambda/\lambda)} \frac{(\lambda/\lambda)^j}{j^!}$, the expected cost to the algorithm for any iteration $i$ is
$(1/e)\sum_j j^2 (1/j!)=O(1)$. 
\item Fix any iteration $i$.  By paying $\Omega(b^2_i)$, the adversary can post $b_i$ bad jobs early on, thereby forcing each good job to pay a fee of at least $b_i+1$.  The algorithm also pays a service cost of $b_i+1$. \vspace{-3pt}
\end{enumerate}

\begin{figure*}[t!]
    \centering
    \includegraphics[width =\textwidth]{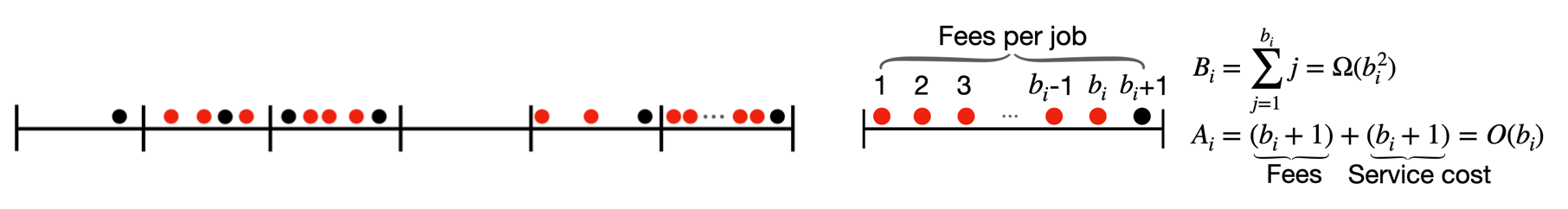}
    \caption{(Left) \Linear with iterations of length $1/\lambda$, using the estimator from the zero-latency example in Section~\ref{s:scenarios}.  Black circles are good jobs, red circles are bad.  Since the estimation gap is small, each iteration contains about $1$ good job.  (Right) Blowup of the last iteration, with $b_i$ bad jobs followed by $1$ good job.  The adversary's cost is $B_i = \Omega(b^2_i)$ and the algorithm's cost is $A_i = O(b_i)$, as described in Section~\ref{s:scenarios}.}
    \label{f:estimator-example} 
\end{figure*}

\noindent
\textbf{Asynchronous.} With communication latency between the server and the clients, jobs may submit multiple fees over time.  If the submitted solution hardness is less than the \currentprice,  the job request is denied.  When this happens, we say that the corresponding job is \defn{bounced} (see Figure~\ref{fig:est-example2}). Alternatively, if the fee exceeds the price, then the job overpays.

Now, assume our example estimator, so that iterations end at the same time as the zero-latency case above.
Further, assume a simple distribution for jobs where exactly $1$ job is generated per iteration at some uniformly distributed time.

\begin{wrapfigure}[17]{r}{.5\textwidth}
\centering
\vspace{-25px}
%\begin{figure}
%\begin{center}
\includegraphics[width=0.45\textwidth]{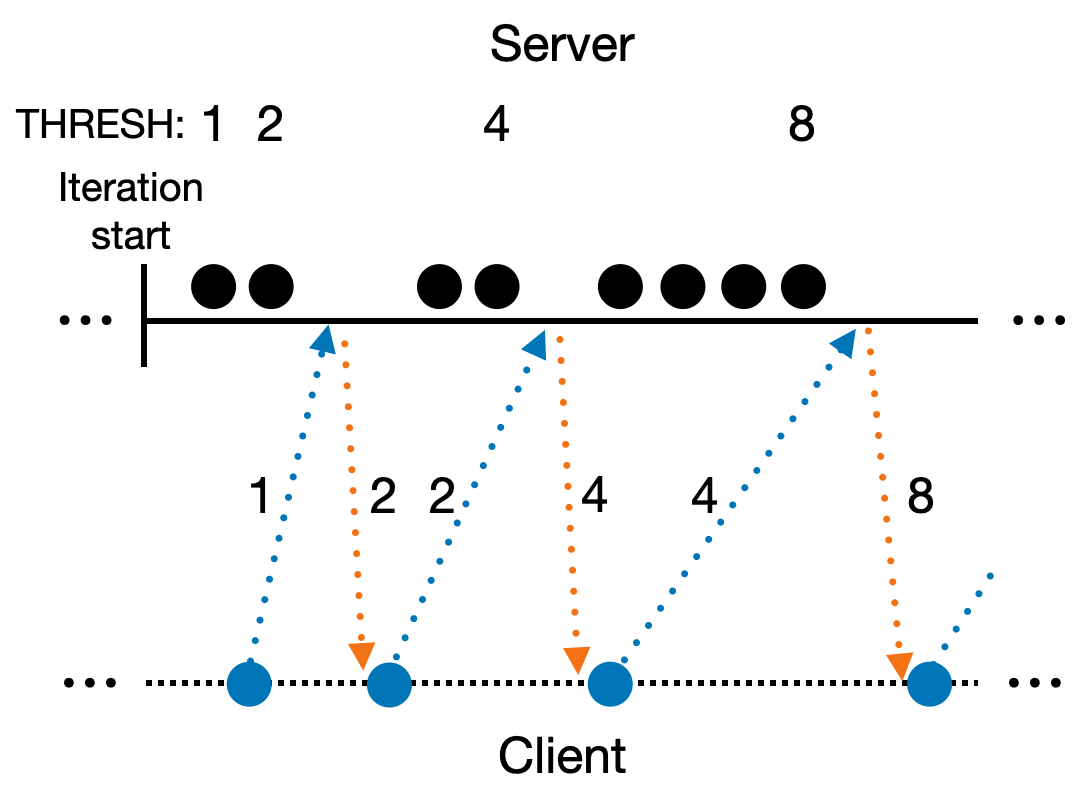}
%\end{center}
\vspace{-8pt}
\caption{The blue circle is a job that is repeatedly bounced, the black circles are other serviced jobs.    Arrows are communication from the server to the client of the $\thresh$ value, and from the client to the server of fees.  The blue job is bounced $3$ times and pays $1+2+4+8$.}\label{fig:est-example2}
%\end{figure}
\vspace{0cm}
\end{wrapfigure}

Note that $M$ is determined by the maximum latency $\Delta$, and $\lambda$. If the latency $\Delta$ is much larger than $1/\lambda$, then the adversary can delay messages so that $M=\Delta/\lambda$ good jobs, bounced from the previous $\Delta/\lambda$ iterations, are received in the current iteration.  The first $M/2$ of these jobs are serviced, driving up the value of $\thresh$ to at least $M/2$.  During this time, the remaining jobs are bounced up to $\log (M/2)$ times.  Thus, the algorithm's total cost is $\Omega((g/M)(M^2)) = \Omega(g M)$.

For these jobs $i=2,...,(M-1)$,  the adversary can cause the $i$-th job to bounce $\lfloor{\lg (M-i)}\rfloor$ times with payments of $2^{j}$, for $j=0,..., \lfloor{\lg (M-i)}\rfloor$ in the current iteration. Thus, the total cost to the algorithm is $\Omega((g/M)(M^2)) = \Omega(g M)$.     

%%%%%%%%%%%%%%%%%%%%%%%%%%%%%%%%%%%%%%%%%%%%%%
%%%%%%%%%%%%%%%%%%%%%%%%%%%%%%%%%%%%%%%%%%%%%%
%%%%%%%%%%%%%%%%%%%%%%%%%%%%%%%%%%%%%%%%%%%%%%

%\section{Analysis}

%In this section, we provde the analysis for \Linear (Section~\ref{s:lin-anal}) and \LinPow (Section~\ref{s:linPow-anal}), and our lower-bound argument (Section \ref{s:lower}).  

%%%%%%%%%%%%%%%%%%%%%%%%%%%%%%%%%%%%%%%%%%%%%%
%%%%%%%%%%%%%%%%%%%%%%%%%%%%%%%%%%%%%%%%%%%%%%
%%%%%%%%%%%%%%%%%%%%%%%%%%%%%%%%%%%%%%%%%%%%%%

\section{Analysis of \Linear}\label{s:lin-anal}

In this section, we formally analyze \Linear.  To give intuition about the choice of cost function, we analyze a general class of algorithms that define iterations like \Linear: algorithms that charge the $i$-th job in the iteration a fee of $i^{\algExp}$, for some positive $\algExp$.  Our analysis shows that \Linear, which sets $\algExp = 1$, is asymptotically optimal across this class of algorithms.  

Let $b$ be the total number of bad jobs, and recall that $g$ is the number of good jobs.   
We note that all good jobs pay a cost of at most $1$ before they learn the current value of \thresh, which incurs a total cost of $O(g)$.  This is asymptotically equal to the total service cost.  Thus, for simplicity, in this section, we only count the costs for the fees submitted when the jobs are serviced.  Throughout, we let $\log(\cdot)$ denote the logarithm base $2$.

%\begin{restatable}{lemma}{upperIter}
\begin{lemma} \label{l:upperIter}
    The number of iterations is no more than $\estGap (g + 1)$.
\end{lemma}
%\label{l:upperIter}
%\end{restatable}
\begin{proof}
    Let $\ell$ be the total number of iterations.  For all $j$, $1 \leq j \leq \ell$, let $\intvl_j$ be the time interval spanning the $j$-th iteration, i.e the interval $(t_0,t_1]$ where $t_0$ is the time the $j$-th iteration begins and $t_1$ is the time the $j$-th iteration ends.

By the algorithm, $\hat{g}({\intvl_j}) \geq 1$ for all $j$.  Let $\intvl_{tot} = \cup_{1 \leq j \leq \ell} \intvl_j$.  Using the additive property of $\hat{g}$, we get $\hat{g}(\intvl_{tot}) = \sum_{1 \leq j \leq \ell} \hat{g}(\intvl_j) \geq \ell$.  But by our definition of $\estGap$, we know that $\hat{g}(\intvl_{tot}) \leq \estGap (g(\intvl_{tot}) + 1) = \estGap(g + 1)$.  Thus, $\ell \leq \estGap(g+1)$.
\end{proof}

\begin{lemma} \label{l:upperGood}
    The number of good jobs serviced in any iteration is no more than $\estGap(\estGap+1)+1$.
\end{lemma}
\begin{proof}
Fix some iteration that starts at time $t_0$ and ends at time $t_1$.  Let $\intvl = [t_0,t_1]$ be the interval associated with that iteration.  Let $\intvl' = [t_0,t_1-\epsilon]$ where $\epsilon \geq 0$ is half the amount of time elapsed between $t_1$ and the time the server receives the second to last job request of the iteration.    First, by construction of $\intvl'$, we have
\begin{equation} \label{e:1}
   g(\intvl') \geq g(\intvl) - 1 
\end{equation}

Then, by the $\estGap$ estimation gap for $\hat{g}$ on the input, we know that $g(\intvl') \leq \estGap (\hat{g}(\intvl') + \estGap)$.  Plugging in the fact that $\hat{g}(\intvl') < 1$, which holds by the properties of \Linear, we get:
\begin{equation} \label{e:2}
 g(\intvl') < \estGap + \estGap^2
\end{equation}
Putting together equations~\ref{e:1} and~\ref{e:2} we get that $g(\intvl) \leq \estGap + \estGap^2 + 1$.
\end{proof}

\begin{lemma} \label{l:lowerIter}
    The number of iterations is at least $\max\left(1, \frac{g-\estGap^2}{3\estGap^4}\right)$.  
\end{lemma}
\begin{proof}
    Let $\intvl_j$ be the interval spanning iteration $j$ for $1 \leq j \leq \ell$, where $\ell$ is the total number of iterations. By Lemma \ref{l:upperGood}, $g(\intvl_j) \leq \estGap(\estGap + 1)+1$ for all $j$.  By the estimation gap for $\hat{g}$ on the input, we get $\hat{g}(\intvl_j) \leq \estGap g(\intvl_j) + \estGap$.  Combining these two inequalities, we get $\hat{g}(\intvl_j) \leq \estGap^3 + \estGap^2 + 2\estGap$.
    
    Let $\intvl_{tot} = \cup_{1 \leq j \leq \ell} \intvl_j$.  Using the additive property of $\hat{g}$, we get 
    $$\hat{g}(\intvl_{tot}) = \sum_{1 \leq j \leq \ell} \hat{g}(\intvl_j) \leq (\estGap^3 + \estGap^2 + 2\estGap) \ell.$$ 
    
    Also, by the estimation gap for $\hat{g}$ on the input, we have $\hat{g}(\intvl_{tot}) \geq g(\intvl_{tot})/\estGap - \estGap$.  Thus, $$\ell \geq (g/\estGap - \estGap)/(\estGap^3 + \estGap^2 + 2\estGap) = (g-\estGap^2)/\estGap(\estGap^3+ \estGap^2 + 2\estGap)$$
which proves the lemma.
\end{proof}

\noindent Recall that $\AdvTotal$ is the cost to the adversary.
\begin{lemma} \label{l:costAdv}
For any positive $\algExp$,
    $\AdvTotal \geq \frac{(b-\estGap (g+1))^{\algExp + 1}}{(\algExp+1)(\estGap (g+1))^\algExp}$.
\end{lemma}
%an exchange argument shows that 
%: Assume $\estGap g > 0$
\begin{proof}
The cost to the adversary is minimized when the bad jobs in each iteration occur before the good jobs.  Furthermore, this cost is minimized when the bad jobs are distributed as uniformly as possible across iterations. To see this, assume $j$ jobs are distributed  non-uniformly among $\ell$ iterations. Thus, there exists at least one iteration $i_1$ with at least $\ceil{j/\ell}+1$ bad jobs, and also at least one iteration $i_2$ with at most $\floor{j/\ell}-1$. Since the cost function $i^{\algExp}$ is monotonically increasing, moving one bad job from iteration $i_1$ to iteration $i_2$ can only decrease the total cost.

To bound the adversarial cost in an iteration, we use an integral lower bound.  Namely, for any $x$ and $\algExp$, $\sum_{i=1}^x i^{\algExp} \geq 1 + \int_{i=1}^x i^\algExp \,di \geq \frac{x^{\algExp+1}}{\algExp+1}$, where the last inequality holds for $\algExp \geq 0$.  If the number of iterations is $\ell > 0$, we let $x = \floor{b/\ell}$ in the above and bound the total cost of the adversary as:
\begin{align*}
\AdvTotal & \geq \frac{\ell}{\algExp+1} {\floor{\frac{b}{\ell}}^{\algExp+1}}\\
& \geq \frac{\ell}{\algExp+1} \left(\frac{b}{\ell} - 1 \right)^{\algExp+1}\\
& \geq \frac{b-\ell}{\algExp+1} \left(\frac{b}{\ell} - 1 \right)^{\algExp}\\
%& \geq \frac{b-\estGap (g+1)}{\algExp+1} \left(\frac{b}{\estGap (g+1)} - 1 \right)^{\algExp}\\
& \geq \frac{b-\estGap (g+1)}{\algExp+1} \left(\frac{b-\estGap (g+1)}{\estGap (g+1)}\right)^{\algExp}\\
& \geq \frac{(b-\estGap (g+1))^{\algExp + 1}}{(\algExp+1)(\estGap (g+1))^\algExp}
\end{align*}

The second step holds since $(b/\ell)-1 = (b-\ell)/\ell$.
The fourth step holds since, by Lemma~\ref{l:upperIter}, $\ell \leq \estGap (g+1)$.    
\end{proof}

Recall that the cost to the total algorithm is the sum of the fees paid by clients and the service cost. In the following, let {\boldmath{$\AlgTotal_{F}$}} denote the fees paid by the clients.

\begin{lemma} \label{l:CostAlgRB}
    For $\algExp \geq 1$, $$\AlgTotal_{F} = O\left(\estGap^{2\algExp} g +  \estGap^2 b^{\algExp} \right). $$
 
    For $0 \leq \algExp < 1$, $$\AlgTotal_{F} = O\left( \estGap^{2\algExp + 4}g + \frac{\estGap^{4\algExp} b^{\algExp + 1}}{\max\{1,(g-\estGap^2)^{\algExp}\}}\right).$$
\end{lemma}
\begin{proof}
 Let $\ell$ be the last iteration with some positive number of good jobs, and for $i \in [1,\ell]$ let $g_i$ and $b_i$ be the number of good and bad jobs in iteration $i$, respectively.  Then, the fees paid in an iteration are maximized when all good jobs come at the end of the iteration.  For a fixed iteration $i$, the cost is:
$\sum_{j=1}^{g_i} (b_i + j)^\algExp$.  Then, the total fees paid, $\AlgRB$ is:
$$\sum_{i=1}^\ell \sum_{j=1}^{g_i} (b_i + j)^\algExp$$
Note that for all $1 \leq i \leq \ell$, $g_i \leq \estGap^2 + \estGap +1$, by Lemma~\ref{l:upperGood}.  For simplicity of notation, in this proof, we let $\beta = \estGap^2 + \estGap + 1$; and note that for all $i \in [1,\ell]$, $g_i \leq \beta$. 

We use exchange arguments to determine the settings for which $\AlgRB$ is maximized subject to our constraints. Consider any initial settings of $\overrightarrow{g}$ and $\overrightarrow{b}$. We assume, without loss of generality, that the $b_i$ values are sorted in decreasing order, since if this is not the case, we can swap iteration indices to make it so, without changing the function output.

\medskip
\noindent\textbf{Good jobs always packed left.} We first show that $\AlgRB$ is maximized when the good jobs are ``packed left'': $\beta$ good jobs in each of the first $\floor{g/\beta}$ iterations, and $g \bmod \beta$ good jobs in the last iteration. 

To see this, we first claim that,  without decreasing $\AlgRB$, we can rearrange good jobs so that the $g_i$ values are non-increasing in $i$.  In particular, consider any two iterations $j,k$ such that $1 \leq j < k \leq \ell$, where $g_j < g_k$.  We move the last $g_k - g_j$ good jobs in iteration $k$ to the end of iteration $j$. This will not decrease $\AlgRB$ since the cost incurred by each of these moved jobs can only increase, since $b_j \geq b_k$, and our RB-cost function can only increase with the job index in the iteration. Repeating this exchange establishes the claim.

Next, we argue that, without decreasing $\AlgRB$, we can rearrange good jobs so that the $g_i$ values are $\beta$ for all $1 \leq i < \ell$.  To see this, consider the case where there is any iteration before the last that has less than $\beta$ good jobs, and let $j$ be the leftmost such iteration with $g_j<\beta$.  Since $j<\ell$, there must be some iteration $k$, where $j<k$, such that $g_k > 0$.  We move the last good job from iteration $k$ to the end of iteration $j$.  This will not decrease $\AlgRB$ since the cost incurred by this moved job can only increase, since our RB-cost function can only increase with the job index in the iteration.
Repeating this exchange establishes the claim.

\smallskip
\noindent\textbf{Analyzing {\boldmath{$\delta_i$}}.}
Next, we set up exchange arguments for the bad jobs.  For any integer $x$, define $$\delta(x) = \sum_{j=1}^{\beta} (x + 1 + j)^{\algExp} -  \sum_{j=1}^{\beta} (x  + j)^{\algExp}$$
If we consider swapping a single bad job from iteration $k$ into iteration $i$ where $1 \leq i < k \leq \ell$, this results in the following change in the overall cost: $\delta(b_i) - \delta(b_{k}-1)$. 

Let $f: \mathbb R \rightarrow \mathbb R$ such that,  $f(x) = \sum_{j=1}^\beta (x + j)^\algExp$. For any $i \in [1,\ell]$, this function is continuous when $x \in [b_i,b_i+1]$ and differentiable when $x \in (b_i, b_i + 1)$, since it is the sum of a finite number of continuous and differentiable functions respectively. For any $i \in [1,\ell]$, by applying the Mean Value Theorem for the interval $[b_i, b_i + 1]$, we have:
\begin{align*}
\delta(b_i) &= \frac{f(b_i+1) - f(b_i)}{(b_i+1)-b_i} = f'(x_i)  \mbox{~for some~} x_i\in[b_i,b_i+1]. 
\end{align*}
\noindent{}\textbf{Bad jobs packed left when {\boldmath{$\algExp \geq 1$}}.} When $\algExp \geq 1$, $f$ is convex. Hence, $f'(x)$ is non-decreasing in $x$.  Using this fact and the Mean Value Theorem analysis above, we get that $\delta(b_i) \geq f'(b_i)$; and, for any $i < k \leq \ell$, $\delta(b_{k} - 1) \leq f'(b_k)$.  In addition, $f'(b_i) \geq f'(b_k)$ since $b_i \geq b_k$, by our initial assumption that the $b_i$ values are non-increasing in $i$.  Thus, $\delta(b_i) - \delta(b_{k} - 1) \geq 0$.  

This implies that $\AlgRB$ does not decrease  whenever we move a bad job from some iteration with index $k>1$ to iteration $i=1$.  Hence, $\AlgTotal_{F}$ is maximized when all bad jobs occur in the first iteration.  

\medskip
\noindent{}\textbf{Bad jobs evenly spread when \boldmath{$\algExp < 1$}.}
When $\algExp < 1$, $f$ is concave. Hence, $f'(x)$ is non-increasing in $x$. Using this fact and the Mean Value Theorem analysis above, we get that $\delta(b_i) \leq f'(b_i)$; and, for any $1 < k \leq \ell$, $\delta(b_{k}-1) \geq f'(b_k)$. In addition, $f'(b_k) \geq f'(b_i)$ since $b_i \geq b_k$, by our initial assumption that the $b_i$ values are non-increasing in $i$. Thus, $\delta(b_i) - \delta(b_{k}-1) \leq 0$.
 
 This implies that $\AlgRB$ increases whenever we swap bad jobs from an iteration with a larger number of bad jobs to an iteration with a smaller number.  Hence, $\AlgRB$ is maximized when bad jobs are distributed as evenly as possible across iterations.\smallskip

\begin{figure*}[t!]
    \centering
    \includegraphics[width = \textwidth]{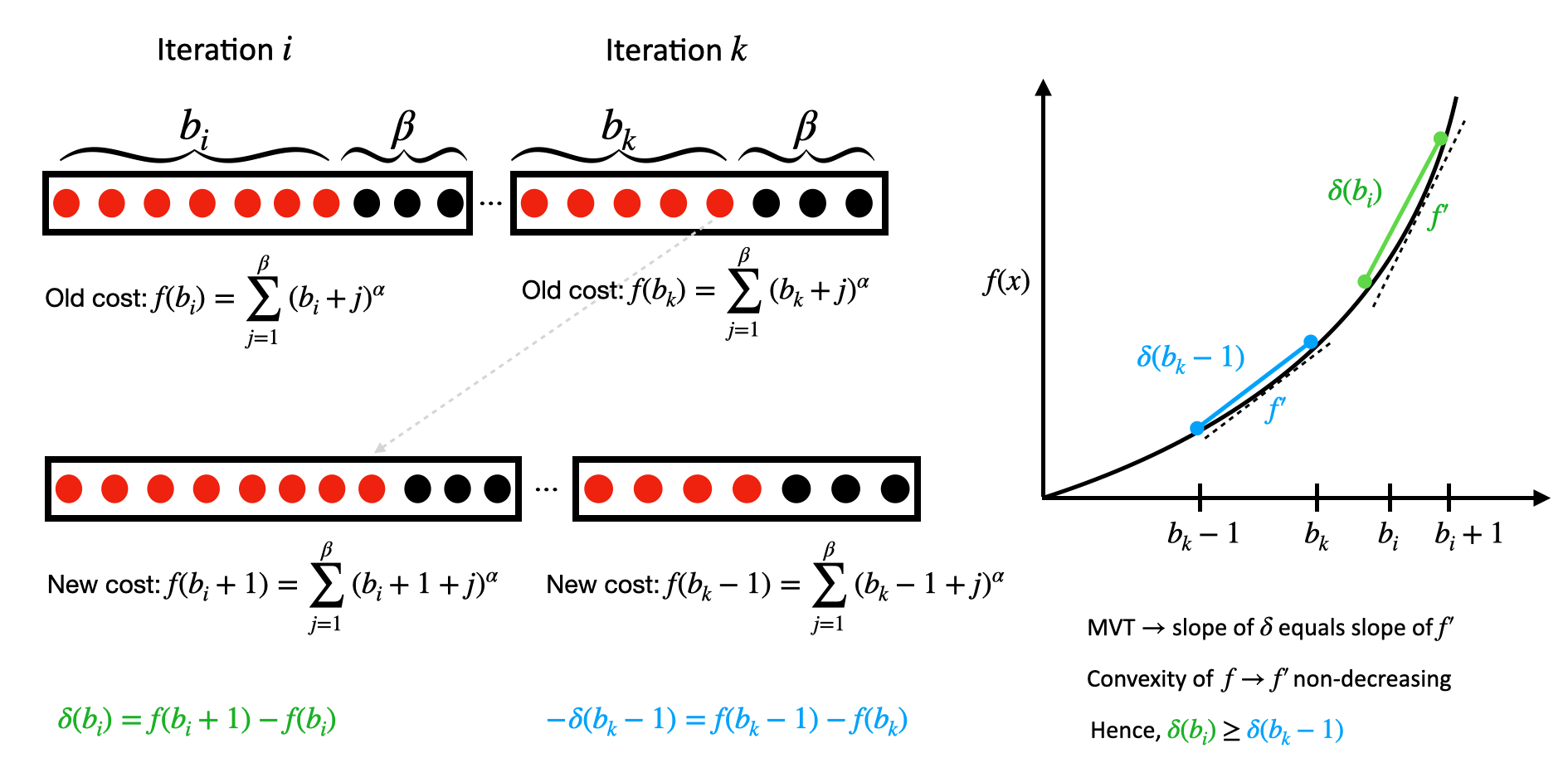}
    \caption{Illustration of an exchange of one bad job (a red ball) from iteration $k$ to iteration $i$, with all good jobs (black balls) remaining in place.  The change in cost for iteration $i$ is $\delta(b_i)$, and change in cost for iteration $k$ is $\delta(b_k - 1)$.  The Mean Value Theorem (MVT) says that for any integer $x$, the value of $\delta(x)$ equals $f'(y)$ for some $y \in [x,x+1]$.  Using the convexity of $f'$ (for $\algExp \geq 1$), this implies that $\delta(b_i) \geq \delta(b_k - 1)$.  This means that the exchange shown can only increase the algorithmic cost.}
    \label{f:mvt}
\end{figure*}

\noindent\textbf{Upper bound on $\AlgTotal_{F}$.}
We can now bound $\AlgRB$ via a case analysis.\smallskip

\noindent{\bf Case {\boldmath{$\algExp \geq 1$}}:} With the above exchange argument in hand, we have:
\begin{align*}
\AlgTotal_{F} & \leq \left(\frac{g}{\beta}\right) \sum_{j=1}^{\beta}j^{\algExp} +  \sum_{j=1}^{\min(\beta,g)}(b+j)^{\algExp}\\
& \leq \left(\frac{g}{\beta}\right) \sum_{j=1}^{\beta}j^{\algExp} +  \min(\beta,g) (b+\beta)^{\algExp}\end{align*}
\begin{align*}
&\leq  2\left(\frac{g}{\beta}\right)  \frac{\beta^{\algExp+1}}{\algExp+1}  +  \min(\beta,g)(b+\beta)^{\algExp}\\
& \leq  2\left(\frac{g}{\beta}\right)  \frac{\beta^{\algExp+1}}{\algExp+1} + \beta(2b)^{\algExp} + g(2\beta)^{\algExp}\\
&= O\left(\beta^{\algExp} g +  \beta b^{\algExp} \right) 
\end{align*}
where the third line follows by upper bounding the sum by an integral, and the fourth line follows from the inequality $(x+y)^{d} \leq (2x)^{d} + (2y)^{d}$ for any positive $x,y$, and $d$, which holds because $x+y \leq 2x$ or $x+y\leq 2y$. Finally,  note that replacing $\beta$ with $\estGap^2+\estGap + 1 \leq 3 \estGap^2$ gives the same asymptotic upper bound in the lemma statement.  \medskip

\noindent{\bf Case {\boldmath{$\algExp < 1$}}:}  With the above exchange argument in hand, we have:
\begin{align*}
\AlgTotal_{F} & \leq \ell \sum_{j=1}^{\beta} (b/\ell + j)^{\algExp} \\
& \leq \ell \sum_{j=1}^{\beta + b/\ell + 1} j^{\algExp} \\
& \leq 2\ell \frac{(\beta + b/\ell + 1)^{\algExp + 1}}{\algExp + 1}
\end{align*}
\begin{align*}
& \leq 2\ell \frac{(2\beta + 2)^{\algExp + 1} + (2b/\ell)^{\algExp+1}}{\algExp + 1}\\
& = O\left( \estGap^{2\algExp + 4}g + \frac{\estGap^{4\algExp} b^{\algExp + 1}}{\max\{1,(g-\estGap^2)^{\algExp}\}}  \right)
\end{align*}
\noindent{} In the above, the last line follows by Lemma~\ref{l:lowerIter}, and noting that $1/\algExp$ is a fixed constant.
\end{proof}

\noindent Note from Lemma~\ref{l:CostAlgRB} that for $b \geq \estGap^2$,  $\AlgRB$ is minimized when $\algExp \geq 1$.  Using this fact, we show $\algExp = 1$ minimizes the total cost of the algorithm as a function of $\AdvTotal$ and $g$.  The following theorem bounds the algorithm's cost.\medskip\smallskip

\begin{lemma}\label{l:costLinear}
The total cost to \Linear is $O\left( \estGap^{5/2} \sqrt{B(g+1)} + \estGap^3 (g+1) \right)$.
\end{lemma}
\begin{proof}
To bound the cost of \AlgTotal, we  perform a case analysis.
\smallskip 
    
\noindent{\bf Case}: {\boldmath{$b \geq 2 \estGap (g + 1)$}}. In this case, $\estGap (g+1) \leq b/2$.  By Lemma~\ref{l:costAdv}, we have:
$$ \AdvTotal  \geq \frac{(b-\estGap (g+1))^{\algExp + 1}}{(\algExp+1)(\estGap (g+1))^\algExp}   \geq \frac{(b/2)^{\algExp + 1}}{(\algExp+1) (\estGap (g+1))^\algExp} = \frac{b^{\algExp+1}}{2^{\algExp+1}(\algExp+1)\estGap^\algExp (g+1)^{\algExp}}$$

\noindent The service cost is $b+g \leq 2b$, since in this case $g \leq b/(2 \estGap) \leq b/2$.  Let $\AlgTotal$ denote the total cost to the algorithm. We now use Lemma~\ref{l:CostAlgRB}; specifically, the case for $\algExp \geq 1$, which gives the smallest worst case total fees paid, across all values of $g$ and $b$ and $\estGap$.  This yields:
$$\AlgTotal \leq  2b+ c(\estGap^{2\algExp} g +  \estGap^2 b^{\algExp})$$
\noindent where $c>0$ is a sufficiently large constant.  Thus we have:
\begin{align*}
    \AlgTotal & \leq 2b + c(\estGap^{2\algExp} g +  \estGap^2 b^{\algExp})  \\
    & \leq 2b + c  (\estGap^{2\algExp} (b/2\estGap) +  \estGap^2 b^{\algExp})  \\
    & \leq 2b + c  (\estGap^{2\algExp-1} b +  \estGap^2 b^{\algExp}) 
\end{align*}
\noindent Where the second line holds since $\estGap (g+1) \leq b/2$ implies that $\estGap g \leq b/2$, or $g \leq b/(2\estGap)$. Note that the above is minimized over the range $\algExp \geq 1$, when $\algExp = 1$.   In this case, we have: 
$$\AlgTotal  \leq 2b + c \left(\estGap b + \estGap^2 b \right) \leq (c\estGap +\estGap^2 + 2)b = O \left( \estGap^{5/2} \sqrt{B(g+1)}\right).$$
%\begin{align*}
%    \AlgTotal & \leq 2b + c \left(\estGap b + \estGap^2 b \right)\\
%     & \leq (c\estGap +\estGap^2 + 2)b\\
%     & = O \left( \estGap^{5/2} \sqrt{B(g+1)}\right).
%\end{align*}
In the equality above, for $\algExp=1$, $\AdvTotal \geq \frac{b^{2}}{8\estGap (g+1)}$, which implies that $b \leq \sqrt{8\estGap B(g+1)}$.  
\medskip

\noindent{\bf Case:} {\boldmath{$b < 2 \estGap (g+1)$}}. In this case, the service cost is $g+b \leq (2\estGap+1) (g+1)$.  So when $\algExp = 1$, we can bound the total cost as:
\begin{align*}
    \AlgTotal  & \leq (2\estGap+1) (g+1) + c(\estGap^2 g +  \estGap^2 b)\\
    &\leq (2\estGap+1) (g+1) + c(\estGap^2 g +  \estGap^2 (2 \estGap (g+1)))\\
    & = O(\estGap^3 (g+1))
\end{align*}

%\begin{align*}
%    \AlgTotal & \leq (2\estGap+1) (g+1) + c(\estGap^2 g +  \estGap^2 b)\\
%    & \leq (2\estGap+1) (g+1) + c(\estGap^2 g +  \estGap^2 (2 \estGap (g+1)))\\
%   & = O(\estGap^3 (g+1))
%\end{align*}
\noindent which completes the proof.
\end{proof}

\noindent The communications bounds for \Linear are straightforward.

\begin{lemma}\label{l:bwLinear}
    In \Linear, the server sends $O(g)$ messages to clients, and clients send a total of $O(g)$ messages to the server.  Also, each good job is serviced after an exchange of at most $3$ messages between the client and server
\end{lemma}
\begin{proof}
     Since there is no latency, each client immediately receives the current value of \thresh if its first message to the server has a fee with value too small.  Then, on the second attempt the client attaches a fee with the correct value, and so receives service.  Hence, each good job is serviced with at most $2$ messages from client to server and at most $1$ message from server to client.
\end{proof}

\noindent Theorem~\ref{t:Linear} now follows directly from Lemmas~\ref{l:costLinear} and~\ref{l:bwLinear}. 

%%%%%%%%%%%%%%%%%%%%%%%%%%%%%%%%%%%%%%%%%%%%%%%%%%%%
%%%%%%%%%%%%%%%%%%%%%%%%%%%%%%%%%%%%%%%%%%%%%%%%%%%%
%%%%%%%%%%%%%%%%%%%%%%%%%%%%%%%%%%%%%%%%%%%%%%%%%%%%

\section{Analysis of \LinPow} \label{s:linPow-anal}

How much do the total fees paid increase for \LinPow in our harder asynchronous problem variant (recall Section~\ref{s:problem})?  We say that a good job ``overpays'' if the fee submitted when the job is serviced exceeds the price that the server demands.  We bound the total amount overpaid by all good jobs in the asynchronous model.  When compared to \Linear, the adversary's cost in \LinPow decreases by at most a factor of $2$; and the total amount spent by good jobs that do not overpay increases by at most a factor of $2$.  Thus, overpayment is the only possible source of asymptotic increase in algorithmic costs.  

Thus we upper bound the overpayment amount.  To do this, we partition time into epochs.  An \defn{epoch} is a contiguous time interval ending with a period of $2 \Delta$ seconds, where the \currentprice during the entire period never increases. 

In the following, for all $i \geq 1$, let $B_i$ be the amount spent by the adversary in epoch $i$.  For  convenience, we define $B_0 = 0$.

\begin{lemma} \label{l:maxRBThresh}
    For all $i \geq 1$, the maximum \currentprice in epoch $i$ is at most $$3 \estGap^2 \max\left(1, 2\sqrt{B_i}\right) \leq 6\estGap^2\sqrt{B_i+1}$$
\end{lemma}
\begin{proof}
Fix an epoch $i \geq 1$ and let $2^x$ be the maximum \currentprice in that epoch.  

If $B_i > 0$, and the epoch consists solely of bad jobs, we know that $B_i \geq 2^{2(x-1)}$, because of the way that our server increases the \currentprice.  Hence, the bad jobs contribute no more than $((1/2) \log B_i) + 1$ to the value $x$, when $B_i > 0$.

In any period where the \currentprice monotonically increases, the number of good jobs is no more than the number of good jobs in an iteration, which is at most $\estGap(\estGap+1)+1 \leq 3 \estGap^2$  by Lemma~\ref{l:upperGood}.  Thus, the total increase in $x$ from good jobs is at most $\log (3 \estGap^2)$.

Putting these together, we get that when $B_i=0$, $x \leq \log 3\estGap^2$ and when $B_i>0$,  $x \leq ( (1/2) \log B_i) + 1 + \log 3\estGap^2$.  Hence the maximum \currentprice   is at most $3 \estGap^2 \max(1, 2\sqrt{B_i}) \leq 6 \estGap^2 \sqrt{B_i + 1}$.  
\end{proof}

\begin{lemma} \label{l:numJobsEpoch}
    For all $i \geq 1$, at most $6M (\log (\estGap (B_i+1)) + 1)$ good jobs arrive in epoch $i$. 
\end{lemma}
\begin{proof}
    Over every $2 \Delta$ seconds, starting from the beginning of the epoch, the \currentprice  must double, or else the epoch will end.  Thus, the epoch lasts at most $2 \Delta (x + 1)$ seconds, where $2^x$ is the maximum \currentprice achieved in that epoch.

    By Lemma~\ref{l:maxRBThresh}, $2^x \leq  6\estGap^2\sqrt{B_i+1}$; taking the log and simplifying, we get $x \leq 3 + 2 \log (\estGap (B_i+1))$.  The lemma follows since over a period of $2 \Delta$ seconds, at most $2M$ good jobs can arrive.
\end{proof}

\begin{lemma} \label{l:nonZeroEpoch}
Any epoch has overpaying jobs only if either the epoch or its preceding epoch spans the beginning of some iteration.
\end{lemma}
\begin{proof}
All jobs entering prior to the last $2 \Delta$ seconds of an epoch are serviced before the end of the epoch, since, by definition, the \currentprice is non-increasing during those $2 \Delta$ seconds.  Thus, all jobs serviced in some fixed epoch either have entered during that epoch or entered in the final $2 \Delta$ seconds of the 
preceding epoch.  

If neither the current epoch nor the preceding epoch spanned the beginning of an iteration, then the \currentprice never resets to $1$ during the lifetime of any of these jobs.  So the \currentprice is non-decreasing during each such job's lifetime.  Thus, none of the jobs can ever have an attached fee that is larger than the \currentprice, and so no overpayment occurs.
\end{proof}

\begin{lemma} \label{l:overpayEpoch}
  For all $i \geq 1$, the total amount that good jobs overpay in epoch $i$ is at most $18M \estGap^2 (\sqrt{B_{i-1}} + \sqrt{B_i} + 1) (\log (\estGap (B_i+1))+2)$.
\end{lemma}
\begin{proof}
    First, note that every good job serviced in epoch $i$ either arrived in epoch $i$: at most $6M (\log (\estGap (B_i+1) + 1)$ such jobs by Lemma~\ref{l:numJobsEpoch}; or arrived in the last $2 \Delta$ seconds of the preceding epoch: at most $2M$ such jobs.  Thus, there are at most $6M (\log (\estGap (B_i+1)) + 2)$ good jobs serviced in epoch $i$.
    
    By Lemma~\ref{l:maxRBThresh}, the maximum \currentprice  in either epoch $i$ or $i-1$ is:
    $$3 \estGap^2 \max\left\{\sqrt{B_{i-1}}, \sqrt{B_{i}}, 1\right\} \leq 3 \estGap^2 \left(\sqrt{B_{i-1}} + \sqrt{B_i} + 1 \right)$$  
 This is the maximum amount that any job serviced in epoch $i$ can overpay.  Multiplying by the number of good jobs serviced, we obtain the bound in the lemma statement. %  bound the total amount overpaid: $$4M (\log (\estGap (B_i+1)) + 2) \cdot 2 \estGap^2 \left(\sqrt{B_{i-1}} + \sqrt{B_i} + 1 \right) \leq 13M \estGap^2 \left(\sqrt{B_{i-1}} + \sqrt{B_i} + 1\right) \log (\estGap (B_i+1))).$$
\end{proof}

\begin{lemma} \label{l:totOverpay}
The total amount overpaid by good jobs in \LinPow is
$$O\left(M \log (\estGap (B+2)) \left(\estGap^{5/2} \sqrt{(g+1)B} + \estGap^3 (g+1) \right)  \right)$$
\end{lemma}
\begin{proof}
Let $X$ be the total amount overpaid by good jobs in \LinPow.  We will give two upperbounds on $X$ and then combine them.

First, by Lemma~\ref{l:maxRBThresh} and the fact that for all $i\geq 1$, $B_i \leq B$, we know that the maximum amount any good job can overpay is at most $6 \estGap^2\sqrt{B + 1}$.  Multiplying by all good jobs, we get:
\begin{equation} \label{e:X1}
    X \leq 6 (g+1) \estGap^2\sqrt{B + 1}
\end{equation}
Next, let $S$ be the set of indices of epochs that have overpaying jobs.  Then, by Lemma~\ref{l:overpayEpoch}: 

\begin{align*}
X & \leq \sum_{j \in S} 18M \estGap^2 (\sqrt{B_{i-1}} + \sqrt{B_i} + 1) (\log (\estGap (B_i+1))+2)\\
& \leq 18 M \estGap^2 (\log (\estGap (B + 1))+2) \sum_{j \in S} (\sqrt{B_{j-1}} + \sqrt{B_j} + 1)\\
& \leq 18 M \estGap^2 (\log (\estGap (B + 1))+2) \left( |S| + \sum_{j \in S} \sqrt{B_{j-1}} + \sum_{j \in S} \sqrt{B_j} \right)\\
& \leq 18 M \estGap^2 (\log (\estGap (B + 1)) + 2) \left(|S| + 2 \sqrt{B} \left(\min(\sqrt{S}, \sqrt{B})\right)\right)
\end{align*}

In the above, the fourth step holds by noting that the sum of the $B_{j-1}$ and $B_j$ terms are both at most $B$, applying Cauchy-Schwartz to both sums, and noting that the sum of the $\sqrt{B_i}$ can never be larger than $B$.  

Next, noting that $|S| \leq 2\estGap (g + 1)$ by Lemmas~\ref{l:upperIter} and \ref{l:nonZeroEpoch}, we have for some constant $C$:
\begin{equation} \label{e:X2}
    X \leq  C M \estGap^{2} \left(\estGap (g+1) + \sqrt{B} \left( \min(\sqrt{\estGap (g+1)}, \sqrt{B}  \right) \log (\estGap (B+2))) \right)
\end{equation}

Putting together equations~\ref{e:X1} and~\ref{e:X2}, we have

\begin{align*}
X &= O \left(\estGap^2 \sqrt{B+1} \min \left(g+1, M \log(\estGap(B+2)) \left(\min(\sqrt{\estGap(g+1)},\sqrt{B+1})  \right) \right) \right)
\end{align*}
\end{proof}

\noindent Lemmas~\ref{l:linpow-total},~\ref{l:max-time}, and~\ref{l:message-complexity} below complete the proof of our second main result, Theorem~\ref{t:LinPow}.

\begin{lemma} \label{l:linpow-total}
  \LinPow has total cost: $$\tilde{O} \left(\estGap^3 \left(\sqrt{B+1} \min \left(g+1, M \sqrt{g+1},M \sqrt{B+1} \right) + (g+1) \right) \right).$$
\end{lemma}
\begin{proof}
Fix the order in which jobs are serviced by \LinPow.  The \currentprice  when a job is serviced in \LinPow will differ by at most a factor of $2$ from \Linear because of the fact that \LinPow uses powers of $2$.  Thus, the total amount spent by the adversary decreases by at most a factor of $2$, and the total amount spent by the algorithm on all jobs that do not overpay increases by at most a factor of $2$.   It follows that the asymptotic results in Theorem~\ref{t:Linear} hold, except for the costs due to overpayment by the good jobs.

So, we can bound the asymptotic cost of the algorithm as the cost from \Linear plus the cost due to overpayment.  The cost in the theorem statement is the asymptotic sum of the costs from Theorem~\ref{t:Linear} and from Lemma~\ref{l:totOverpay}.  This is:
%\begingroup\makeatletter\def\f@size{8.7}\check@mathfonts
\begin{align*}
&O\mbox{\Large (}\estGap^2 \sqrt{B+1} \min \mbox{\Large (}g+1, M \log(\estGap(B+2)) \mbox{\Large (}\min(\sqrt{\estGap(g+1)},\sqrt{B+1})) \mbox{\Large )} \mbox{\Large )} \\
& + \estGap^{5/2}\sqrt{B(g+1)} + \estGap^3(g+1) \mbox{\Large )}
 \end{align*}
% \endgroup
which is 
$$ 
\tilde{O} \left(\estGap^3 \left(\sqrt{B+1} \min \left(g+1, M \sqrt{g+1},M \sqrt{B+1} \right) + (g+1) \right) \right)
$$

%$$ 
%\tilde{O} \left(\estGap^3 \left(\sqrt{B+1} \min \left(g+1, M \sqrt{(g+1)},M \sqrt{B+1})) \right) + (g+1) \right) \right)
%$$
\noindent which completes the argument.
\end{proof}

\begin{lemma}\label{l:max-time}
    The maximum time for any good job to start service is $\Delta(4\log (\estGap (B+1)) + 4)$ seconds.
\end{lemma}
\begin{proof}
Fix any good job.  Every time the job is bounced, the \currentprice sent from the server to the corresponding client at least doubles.  The client receives the new \currentprice from the server within $\Delta$ seconds.

How many times can the \currentprice change?  If $B>0$, the number of times that bad jobs can double the \currentprice is at most $(1/2) \log B$.  By Lemma~\ref{l:upperGood}, the number of good jobs in an iteration is at most $\estGap(\estGap+1) + 1 \leq 3 \estGap^2$.  Hence, the total number of times the \currentprice can change for $B=0$ is at most $2 + 2 \log \estGap$; and for $B>0$, at most $2 + 2 \log \estGap + (1/2) \log B$.  In both cases, this is at most $2 + 2 \log \estGap + (1/2) \log (B+1) \leq 2 + 2 \log (\estGap (B+1))$.

Every time the job is bounced, the latency increases by at most $2\Delta$ seconds: $\Delta$ seconds to send the job to the server, and $\Delta$ seconds to receive the new \currentprice  from the server.  When the job is serviced it takes at most $\Delta$ seconds: the time to send the job to the server.  Thus, the total time for the job to be serviced is at most $2\Delta (2\log( \estGap (B+1)) + 2) + \Delta \leq \Delta(4\log (\estGap (B+1)) + 4)$.
\end{proof}

The following lemma bounds the number of messages sent from the server to clients and from clients to the server.  Note that the adversary can cause the server to send a message to it for every request that it sends.

\begin{lemma}\label{l:message-complexity}
    In \LinPow the server sends $O(g \log(B+(3 \estGap^2)))$ messages to the clients and the clients send a total of $O(g \log (B+3\estGap^2))$ messages to the server. 
\end{lemma}
\begin{proof}
Every time the server sends a message to the client, the \currentprice sent to that client must at least double. How often can this happen?  By Lemma~\ref{l:upperGood} the number of good jobs in any iteration is at most $\estGap(\estGap+1) + 1\leq 3 \estGap^2$.  So, the total number of changes to the \currentprice in any iteration is $\log(B+(3 \estGap^2))$.  So, the maximum \currentprice ever achieved is at most $2^{\log(B+(3\estGap^2))}$.  Hence, the maximum number of messages sent to any client is $O(\log(B+(3 \estGap^2)))$.  It follows that the total number of messages sent to clients is $O(g \log(B+(3 \estGap^2)))$

Every time a client sends a new message to the server, the fee attached to that message doubles.  The maximum \currentprice set by the server over all iterations is at most $\log (B+3\estGap^2)$, so a good job will send at most that number of messages. Hence, all $g$ good jobs send a total of $O(g \log (B+3\estGap^2))$ messages to the server.  
\end{proof}

%%%%%%%%%%%%%%%%%%%%%%%%%%%%%%%%%%%%%%%%%%%%%%
%%%%%%%%%%%%%%%%%%%%%%%%%%%%%%%%%%%%%%%%%%%%%%
%%%%%%%%%%%%%%%%%%%%%%%%%%%%%%%%%%%%%%%%%%%%%%

\section{Lower Bound Proof}\label{s:lower}

In this section we give a lower bound for randomized algorithms for our problem zero-latency.  This lower bound also holds directly for our harder asynchronous variant of our model.

\subsection{Our Adversary} \label{s:adv}
Let $\estGap$ be any positive integer, $g_0$ be any multiple of $\estGap$, and $n$ be any multiple of $g_0 \estGap$.  To prove our lower bound, we choose both an input distribution and also a specific estimator.

\medskip
\noindent{}
\textbf{Distribution.}
There are two possible distributions of $n$ jobs.  In both distributions, the jobs are evenly spread in a sequence over $T$ seconds for any value $T$;  in particular, the $i$-th job occurs at time $T (i/n)$.

In distribution one, the number of good jobs is $g_1 = \estGap g_0$.  In distribution two, it is $g_2 = g_0/\estGap$.  Our adversary chooses one of these two values, each with probability $1/2$.

Then, for $i \in \{1,2\}$, we partition from left to right the sequence of jobs into contiguous subsequences of length $n/g_i$. In each partition, we set one job, selected uniformly at random, to be good and the remaining jobs to be bad.   

\medskip
\noindent{}\textbf{Estimator.}  For any interval $\intvl$, let $\numjobs(\intvl)$ be the number of (good and bad) jobs overlapped by $\intvl$, i.e. $\numjobs(\intvl) = |\intvl \cap \{k (T/n): k \in \mathbb{Z}^+ \}|$. Then 
$$\hat{g}(\intvl) = (g_0/n)\numjobs(\intvl)$$

It is easy to verify that the above function has the additive property, and so is an estimator.  Moreover, this estimator has two important properties.  First, the estimation gap is always at most $\gamma$ for either input distribution (see Lemma~\ref{l:estim}).  Second, it is independent of the choice of distribution one or two.

\subsection{Analysis}

We first show that this estimator has $\estGap$ estimation gap for both distributions. 

\begin{lemma} \label{l:estim}
For both distributions, for any interval $\intvl$, we have:
$$g(\intvl)/\estGap - \estGap \leq \hat{g}(\intvl) \leq \estGap g(\intvl) + \estGap$$
\end{lemma}
\begin{proof}
    Let $g_{\max} = \max \{ g_1, g_2 \} = \estGap g_0$ and $g_{\min} = \min \{ g_1, g_2\} = g_0/\estGap$, and fix any interval $\intvl$.  Then, for each distribution we have the following based on the observation that there is one good job in every partition of size $n/g_i$ for $i \in \{1,2\}$
\begin{align}
g(\intvl) \leq \lceil \numjobs(\intvl)/(n/g_i)\rceil \leq  \left(\frac{g_{\max}}{n}\right)\numjobs(\intvl)+1 \leq  \left(\frac{\estGap g_0}{n}\right)\numjobs(\intvl)+1. \label{eqn:upper-g}
\end{align}
\noindent Thus:
\begin{align*}
\left(\frac{1}{\estGap}\right) g(\intvl)-\estGap &\leq \left(\frac{1}{\estGap}\right)\left(\frac{\estGap g_0}{n}\numjobs(\intvl)+1\right) -\estGap\\
&\leq \left(\frac{g_0}{n}\right)\numjobs(\intvl) + \frac{1}{\estGap} - \estGap\\
&\leq \hat{g}(\intvl)
\end{align*}
\noindent In the above, the first line follows from Equation~\ref{eqn:upper-g}, and the last line follows for $\estGap\geq 1$ and by the properties of our estimator. 

For the other direction, note that for any interval $\intvl$, we have:
\begin{align}
g(\intvl) \geq \lfloor \numjobs(\intvl)/(n/g_i)\rfloor \geq \left(\frac{g_{\min}}{n}\right)\numjobs(\intvl)-1 \geq  \left(\frac{g_0}{\estGap n}\right)\numjobs(\intvl)-1.\label{eqn:lower-g}
\end{align}\noindent Thus:
\begin{align*}
\estGap g(\intvl) + \estGap &\geq \estGap\left( \frac{g_0}{\estGap n}\numjobs(\intvl)-1 \right) + \estGap\\
& \geq \left(\frac{g_0}{n}\right)\numjobs(\intvl)\\
& = \hat{g}(\intvl)
\end{align*}

\noindent In the above, the first line follows from Equation~\ref{eqn:lower-g}, and the last line follows from the fact that $\estGap \geq 1$, and the definition of the estimator. Hence, by the above analysis, our estimator has $\estGap$ estimation gap for both distributions. 
\end{proof}

Since the estimator returns the same outputs on both distributions, any deterministic algorithm will set job costs to be the same.  This fact helps in proving the following lemma about expected costs.

 \begin{lemma} \label{l:lbDeterm}
Consider any deterministic algorithm that runs on the above adversarial input distribution and any estimator.  Define the following random variables: $A$ is the cost to the algorithm; $B$ is the cost to the adversary; and $g$ is the number of good jobs in the input.  Then:
     $$E\left(A - \sqrt{Bg\estGap}\right) \geq 0$$
 \end{lemma}
\begin{proof}
By Lemma~\ref{l:estim} the estimator has $\estGap$ estimation gap, and returns the exact same values  whether the number of actual good jobs is $g_1$ or $g_2$.  Thus, the server can set the price based only on the job index.  So, for each index $i \in [1,n]$, let $c_i$ be the price for the $i$-th job; and let $C = \sum_{i=1}^n c_i$.  Let $E_1(A)$ (resp. $E_1(B)$) be the expected cost to the algorithm (resp. the expected adversarial total cost) for distribution $1$; and $E_2(A)$, $E_2(B)$ be analogous expected costs for distribution $2$. 

To show $E(A - \sqrt{Bg\estGap}) \geq 0$, we will show that $E(A)/\sqrt{E(B)g} \geq \sqrt{\estGap}$, and then cross-multiply and use Jensen's inequality.   
If we define a random variable $a_i$ that has value $c_i$ if the $i$-th job is good and $0$ otherwise, and note that $A = \sum_{i=1}^n a_i$, then by linearity of expectation, the expected algorithm cost due to fees in distribution $1$ is $\sum_{i=1}^n E(a_i) = C \estGap g_0/n$, with service cost of $n$.  A similar analysis for $E_1(B)$ gives:
\begin{align*}
E_1(A) & = C\estGap g_0/n + n\\
E_1(B) & = \frac{n - \estGap g_0}{n} C. 
\end{align*}
\noindent Similarly, for the second distribution:
\begin{align*}
E_2(A) & = Cg_0/(\estGap n) + n\\
E_2(B) & = \frac{n - g_0/\estGap}{n} C.
\end{align*}

\noindent{}
Hence, we have:
\begin{align*}
\frac{E_1(A)}{\sqrt{E_1(B) g_1}} & = \sqrt{\frac{C \estGap g_0}{n(n-\estGap g_0)}} + \sqrt{\frac{n^3 }{C\estGap g_0(n-\estGap g_0)}}\\
& = \frac{V_1 + n/V_1}{\sqrt{(n-\estGap g_0)}}
\end{align*}
where $V_1 = \sqrt{\frac{C \estGap g_0}{n}}$. Also:
\begin{align*}
\frac{E_2(A)}{\sqrt{E_2(B) g_2}} & = \sqrt{\frac{C(g_0/\estGap)}{n(n-g_0/\estGap)}} + \sqrt{\frac{n^3}{C(g_0/\estGap)(n-g_0/\estGap)}}\\
& = \frac{V_2 + n/V_2}{\sqrt{(n-g_0/\estGap)}}
\end{align*}
\noindent where $V_2 = \sqrt{\frac{Cg_0/\estGap}{n}} = V_1/\estGap$.

Let $f(C)$ be the expected value of the ratio of $E(A)/\sqrt{E(B) g}$ as a function of $C$.  Since the adversary chooses each distribution with probability $1/2$, we have:
\begin{align*}
    f(C) & = 1/2 \left( \frac{V_1 + n/V_1}{\sqrt{n -\estGap g_0}} + \frac{V_1/\estGap + n\estGap/V_1}{ \sqrt{n - g_0/\estGap}}\right) \\
     & \geq \frac{1}{2(\sqrt{n -\estGap/g_0})} \left((1+1/\estGap)V_1 + (1+\estGap) n/V_1)\right) 
\end{align*}

Consider the expression in parenthesis above: $g(V_1) = (1+1/\estGap)V_1 + (1+\estGap) n/V_1$
The derivative of $g(V_1)$ with respect to $V_1$ is $1/\estGap + 1 - (\estGap +1)n/V_1^2$.  From this, we see that $g(V_1)$ is minimized when $V_1 = \sqrt{n\estGap}$.  Plugging this back in to $f$, we get that $f(C) \geq \sqrt{\estGap}$. 
\begin{align*}
    f(C) &\geq \frac{1}{2(\sqrt{n -\estGap/g_0})} \left((1+1/\estGap)V_1 + (1+\estGap) (n/V_1))\right) \\
%    & =  \frac{1}{2(\sqrt{n -\estGap/g_0})} \left( (1+ 1/\estGap) \sqrt{n\estGap} + (1+\estGap)(n/\sqrt{n\estGap})\right)\\
    & =  \frac{1}{2(\sqrt{n -\estGap/g_0})} \left( \sqrt{n\estGap} + \sqrt{n/\estGap} + \sqrt{n/\estGap} + \sqrt{n\estGap} \right)
    \end{align*}
    \begin{align*}
    &= \frac{1}{2(\sqrt{n -\estGap/g_0})} \left(2\sqrt{n\estGap} + 2\sqrt{n/\estGap}\right)\\
    &= \frac{\sqrt{n\estGap} + \sqrt{n/\estGap}}{\sqrt{n -\estGap/g_0}} \\
    & \geq \frac{\sqrt{n\estGap} + \sqrt{n/\estGap}}{\sqrt{n}}\\
    & \geq \sqrt{\gamma}
\end{align*}

Hence, we have proven that the expected value of the ratio $E(A)/\sqrt{E(B) g \estGap}$ is always at least $1$, for our adversarial input distribution, for any deterministic algorithm. 
 Cross multiplying, we get
\begin{align*}
   E(A) & \geq \sqrt{E(B) g \estGap} \\
   & \geq E(\sqrt{B g \estGap})
\end{align*}

where the second line holds by applying  Jensen's inequality~\cite{mcshane1937jensen} to the convex function\footnote{Equivalently, we can directly apply the (less well-known) version of Jensen's inequality for concave functions~\cite{dragomir1992some} to obtain the second line.} $f(x) = -\sqrt{x}$.  Finally, by linearity of expectation, this inequality implies that $E(A - \sqrt{B g \estGap}) \geq 0$, which completes the argument.
\end{proof}

\smallskip
\noindent{\textsc{Theorem}~\ref{t:main-lower}.} {\it
Let $\estGap$ be any positive integer; $g_0$ be any positive multiple of $\estGap$; and $n$ any multiple of $\estGap g_0$. Next, fix any randomized algorithm in our model.  Then, there is an input  with $n$ jobs, $g$ of which are good for $g \in \{ \estGap g_0, g_0/\estGap \}$; and an estimator with estimation gap $\estGap$ on that input. 
Additionally, the randomized algorithm on that input has expected cost: $$E(A) = \Omega\left(E(\sqrt{\estGap B (g+1)}) + \estGap (g+1)\right).$$}

\begin{proof}
By Lemma~\ref{l:lbDeterm}, on our adversarial distribution, any deterministic algorithm has   $E(A) - \sqrt{\estGap E(B) g \estGap} \geq 0$.

We fix values of $g_0$, $n$ and $\estGap$ in our adversarial distribution and define a game theoretic payoff matrix $\mathcal{M}$ as follows.  The rows of $\mathcal{M}$  correspond to deterministic inputs in the sample space used by our adversarial distribution; the columns correspond to all deterministic algorithms over $n$ jobs.  Further, for each input $i$, and algorithm $j$, the matrix entry $\mathcal{M}(i,j)$ is the value $A - \sqrt{B g \estGap}$, where $A$ and $B$ are the costs to the algorithm and adversary respectively, when algorithm $j$ is run on input $i$; and the value $g$ is the number of good jobs in input $i$.

Lemma~\ref{l:lbDeterm} proves that the minimax of this payoff matrix is at least $0$.  So by Von Neumann's minimax theorem~\cite{nikaido1954neumann,yao:probabilistic}, the maximin is also at least $0$.  In other words, via Yao's principal~\cite{yao:probabilistic}, any \emph{randomized} algorithm has some \emph{deterministic} input for which $E(A - \sqrt{\estGap B g}) \geq 0$.  By linearity of expectation:
\begin{align*}
   E(A) & \geq E(\sqrt{\estGap B g}).
\end{align*}

Finally, every algorithm has job service cost that equals $n = \estGap g$, so $E(A) = n \geq \estGap g$.  Adding this to the above inequality gives that $2E(A) \geq E(\sqrt{\estGap B g}) + \estGap g$, or $E(A) = \Omega(\sqrt{\estGap E(B) g} + \estGap g)$.
Since $g$ is positive, this value is: $$\Omega\left(\sqrt{\estGap E(B) (g+1)} + \estGap (g+1)\right).$$  
Since the sample space of our adversarial distribution has $g \in \{\estGap g_0, g_0/\estGap \}$, the worst-case distribution chosen for the maximin will also have $g \in \{\estGap g_0, g_0/\estGap \}$.
\end{proof}

%%%%%%%%%%%%%%%%%%%%%%%%%%%%%%%%%%%%%%
%%%%%%%%%%%%%%%%%%%%%%%%%%%%%%%%%%%%%%
%%%%%%%%%%%%%%%%%%%%%%%%%%%%%%%%%%%%%%

\section{Preliminary Simulation Results}\label{s:simulation}

We conduct simulations for input distributions where there is a large attack and all good jobs are serviced as late as possible; here, our aim is to simply verify the scaling behavior predicted by theory. Our custom simulations are written in Python and the source code is available to the public via GitHub~\cite{chakraborty:code}. 

\subsection{Experiments for \Linear}

In our experiments for $\Linear$, the distribution of good and bad jobs, along with the performance of the estimator, is setup as follows. Let the jobs be indexed as $J_1, J_2, \ldots, J_n$.  Then, let $\intvl_m$ be the time interval containing the single point in time at which the $m$-th job is generated for $1 \leq m \leq n$. 

\begin{itemize}
    \item For $m \in [1,\estGap]$, job $J_m$ is bad and $\hat{g}(\intvl_{m}) = 1$
    \item For $m \in [\estGap + 1, n-\estGap]$, job $J_m$ is bad and $\hat{g}(\intvl_{m}) = 0$
    \item For $m \in [n-\estGap+1, n-1]$, job $J_m$ is good and $\hat{g}(\intvl_{m}) = 0$
    \item For $m = n$, job $J_m$ is good and $\hat{g}(\intvl_{m}) = 1$
\end{itemize}

For all other intervals $\intvl_t$ consisting of single points in time, $t$, $\hat{g}(\intvl_t) = 0$.  Then, the additive property defines the estimator output for all larger intervals.

Thus, the first $\estGap$ jobs are bad, but considered to be good by the estimator (each such job will end an iteration); the next $n-2\estGap$ jobs are bad and considered to be bad by the estimator; the next $\estGap-1$ jobs are good and considered to be bad by the estimator; and the last job is good and considered to be good by the estimator (which ends an iteration).

\begin{figure*}[t]
\hspace{-0pt} 
\includegraphics[width=0.9\textwidth, trim = 0in 0in 0in 0in, clip]{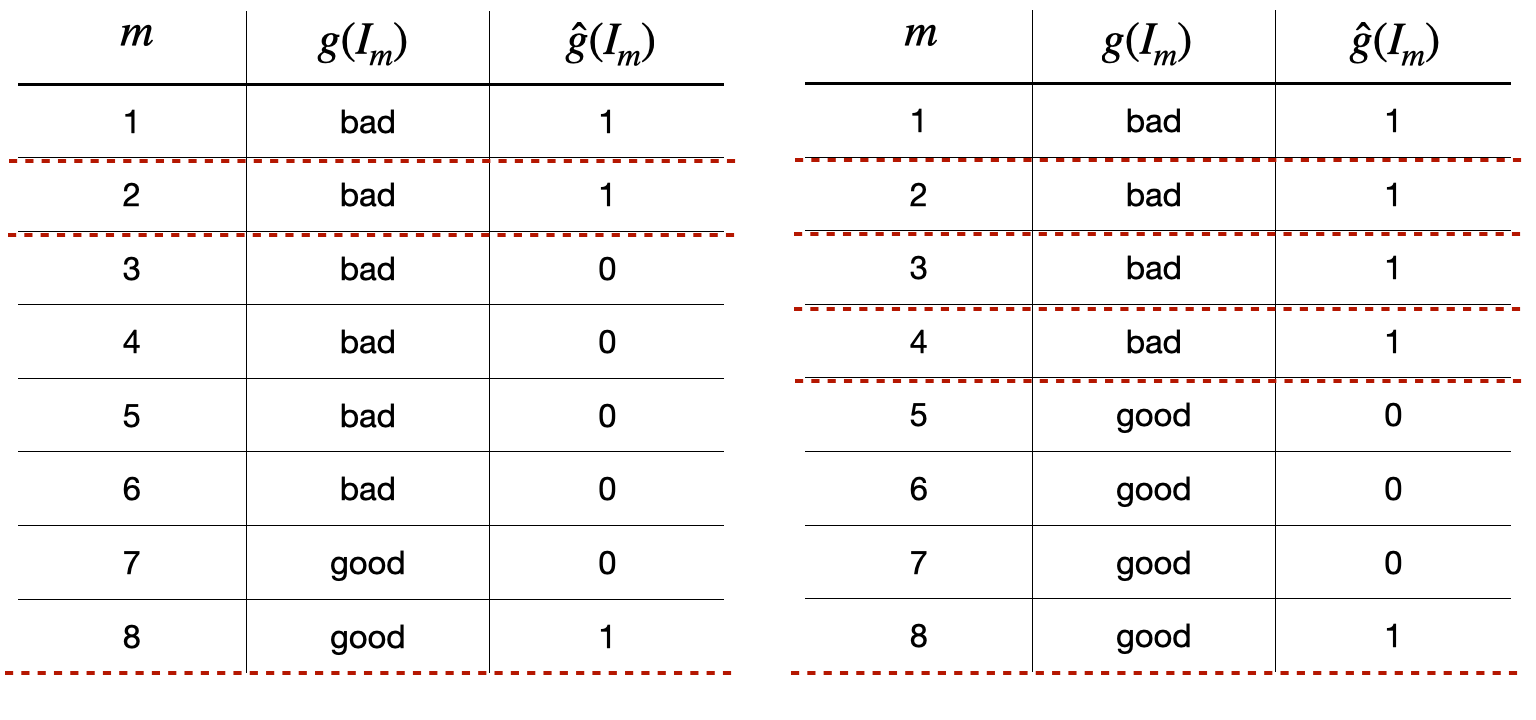} 
\caption{Two small input examples for experiments with $\Linear$ for $n=8$, with $\estGap=2$ (left) and $\estGap=4$ (right). Red dashed lines indicate the end of an iteration.} 
 \label{fig:tables}
\end{figure*}

We highlight that, in the above setup, the estimator errs in a way that favors the adversary and disfavors our algorithm. Given that the first $\estGap$ bad jobs are (erroneously) considered good by the estimator, the adversary pays only $1$ for each such job, since each such job ends an iteration.  Thus, as $\estGap$ grows, the adversary benefits; see the examples for $\estGap=2$ and $\estGap=4$ in Figure~\ref{fig:tables}. Additionally, of the final $\estGap$ (good) jobs, all but the last one are correctly identified as good, which increases the cost to the algorithm.
 \medskip

\noindent{\bf Specific Experiments.}  We conduct two experiments.  For our first experiment,  we  let $n= 10\times2^x$ for integer $x\in[-1,18]$.  Then, we fix $\estGap=1$ and investigate the impact of different values of $\algExp$ (recall Section~\ref{s:lin-anal}); specifically, we let $\algExp=0, 1/2, 1$ and $2$. 

For our second experiment, we set $n= 10\times2^x$, where $x\in[3,18]$.  We fix $\algExp=1$ and investigate the impact of different values of $\estGap$; specifically,  we let $\estGap =1, 2, 4, 8, 16$ and $32$.   The range for $x$ is slightly smaller than in the first experiments, since the smallest power-of-$2$-sized input that holds both good and bad jobs is $10\times2^3=80$ when $\estGap=32$. 
\medskip

%%%%%%%%%%%%%%%%%%%%%%%%%%%%%%%%%%%%%%%%%%%%%%
%%%%%%%%%%%%%%%%%%%%%%%%%%%%%%%%%%%%%%%%%%%%%%
%%%%%%%%%%%%%%%%%%%%%%%%%%%%%%%%%%%%%%%%%%%%%%

\begin{figure*}[t!]
\hspace{-10pt}\begin{subfigure}{0.46\textwidth} 
\includegraphics[width=1.15\textwidth, trim = 0in 0in 0in 0in, clip]{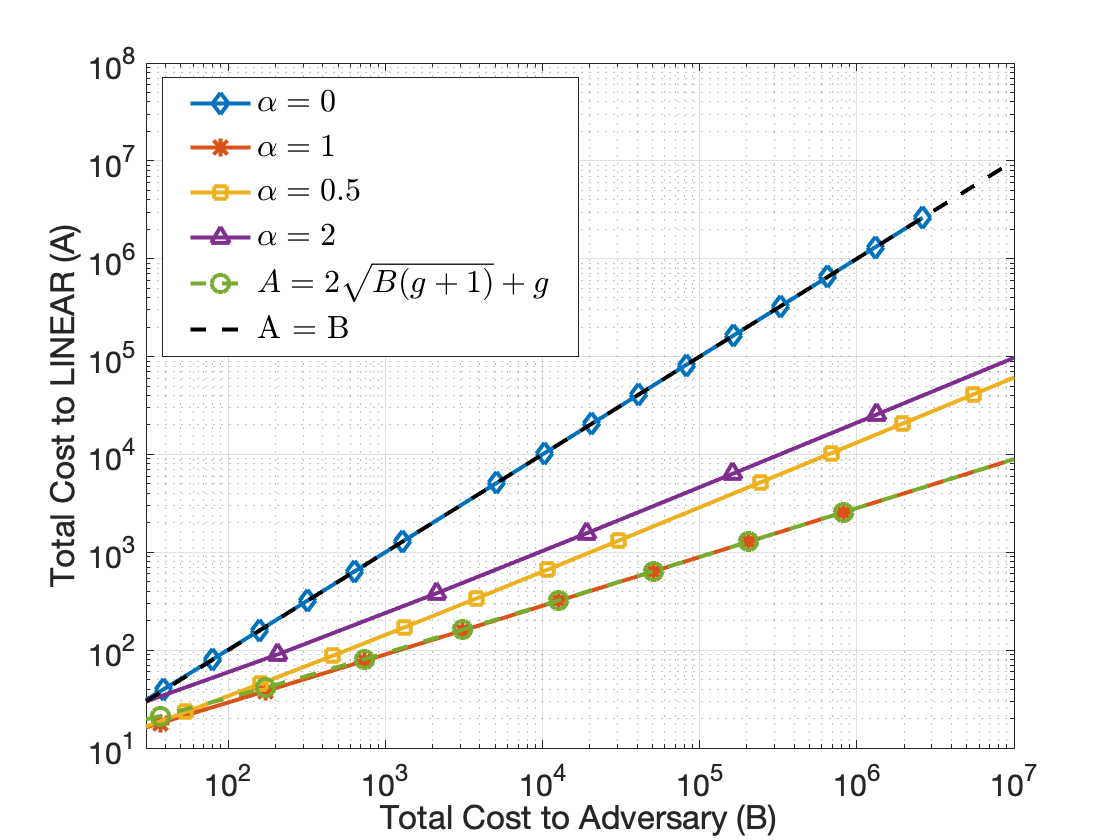} 
\end{subfigure}
\hspace{18pt}\begin{subfigure}{0.46\textwidth} 
\vspace{0pt}\includegraphics[width=1.15\textwidth, trim = 0in 0in 0in 0in, clip]{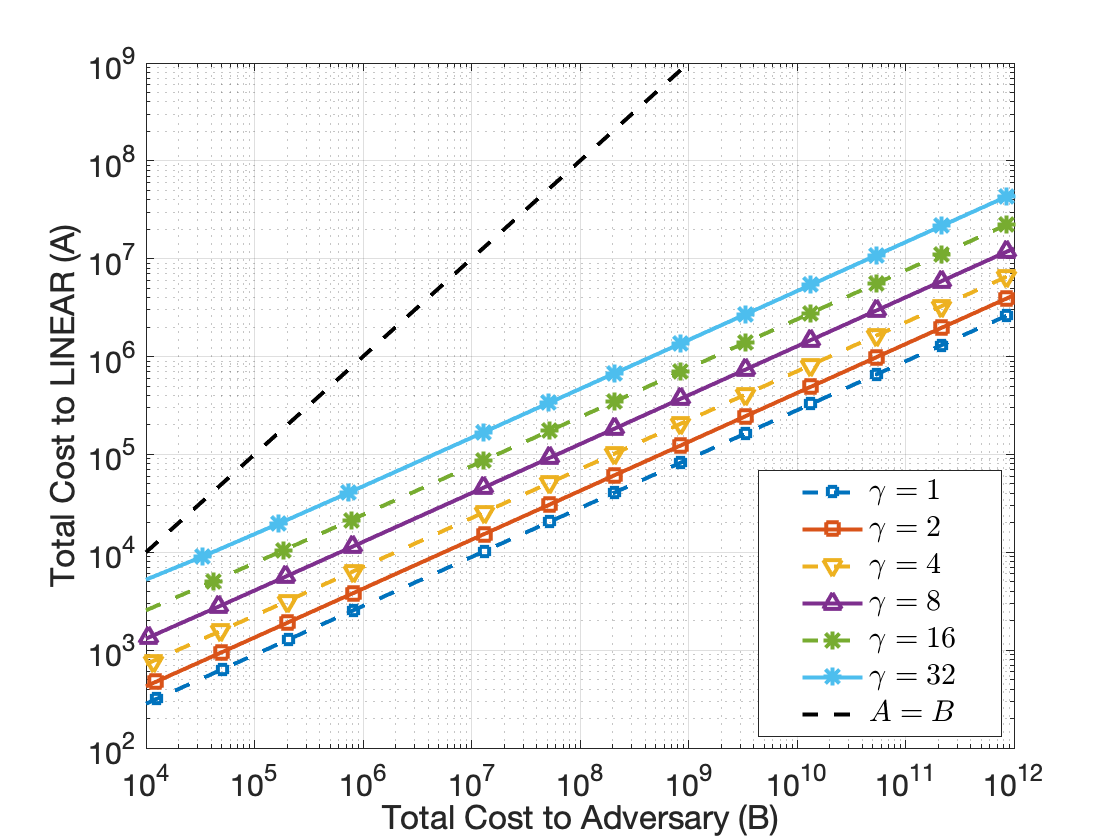} 
\end{subfigure}
\caption{Results for \Linear. (Left) $\estGap=1$; $\algExp$ value varies. (Right) $\algExp = 1$; $\estGap$ value varies.} 
 \label{fig:sim}
\end{figure*}

%%%%%%%%%%%%%%%%%%%%%%%%%%%%%%%%%%%%%%%%%%%%%%
%%%%%%%%%%%%%%%%%%%%%%%%%%%%%%%%%%%%%%%%%%%%%%
%%%%%%%%%%%%%%%%%%%%%%%%%%%%%%%%%%%%%%%%%%%%%%

\noindent{\bf Results for \Linear.} The results of our first experiment are plotted in Figure~\ref{fig:sim} (Left). The algorithm's cost is denoted by $A$, while the adversary's cost is denoted by $B$. The results show scaling consistent with our theoretical analysis. Notably, setting $\algExp=1$ achieves the most advantageous cost for algorithm versus the adversary. We observe that for $\algExp=0, 2$ and $0.5$, the algorithm cost increases by up to a factor of 
 approximately $33.3$, $3.3$, and $1.8$, respectively, relative to  \Linear (i.e., $\algExp=1$) at $B=10^4$. We include the line $A=B$ as a reference, and this closely corresponds to the case for $\algExp=0$ (i.e., each job has cost $1$), since the service cost incurred by the algorithm is close to the number of bad jobs in these experiments. We also include the equation $A=2(\sqrt{\AdvTotal(g+1)}+g)$, which closely fits the line for $\algExp=1$.

The results of our second experiment are plotted in Figure~\ref{fig:sim} (Right).  As expected, when $\estGap$ increases---and, thus, the accuracy of our estimator worsens---the cost ratio of the algorithm to the adversary increases. Specifically,  for $B=10^7$, we see that the values of $\estGap = 2, 4, 8, 16$ and $32$ correspond to a cost ratio that is, respectively, a factor of $1.5$, $2.6$, $4.4$, $8.8$, and $15.5$ larger than than for $\estGap=1$.  Given this, and noting the fairly even spacing of the trend lines on the log-log-scaled plot, the cost appears to scale according to a (small) polynomial in $\estGap$. We note that, despite this behavior, the cost of the algorithm versus the adversary is still  below the $A=B$ line; thus, \Linear still achieves a significant advantage for the values of $\estGap$ tested.  

\vspace{0.5cm}

\subsection{Experiments for \LinPow}

To evaluate a challenging case for \LinPow, we create an input causing many bounced good jobs in the first iteration. From the start of this iteration, we consider disjoint periods of $2\Delta$ seconds.   Each period $i\geq 0$ will have $2^i$ bad jobs serviced, followed by $(i+1)M$ good jobs that are bounced. For example, in the first period, we have $1$ bad job serviced followed by $M$ good jobs; these $M$ good jobs will be bounced, since they start with a solution to a $1$-hard challenge, while the \currentprice has increased to $2$ due to the first bad job. In the second period, two bad jobs are serviced,  followed by $M$ new good jobs, each attaching a fee of $1$.  These are followed by the $M$ good jobs that were bounced in the first period.   All of these $2M$ good jobs will be bounced, since the \currentprice has increased to $4$ due to the two bad jobs.

For each integer $z\in [10, 28]$, we fix a number of bad jobs $b=2^z - 1$, and let $M$ range over values $1, 2, 4,$ and  $8$. Once all $b$ bad jobs have been serviced---which occurs in period $z$--- a final batch of $M$ good jobs arrive with a $1$ hard solution and these are bounced.  However, the other $(z+1)M$ (good) jobs will start being serviced. Specifically, $\estGap$ of these (good) jobs will be serviced, with the first $\estGap-1$ jobs (erroneously) considered bad by the estimator, and $\estGap$-th  job (correctly) considered good; this ends the first iteration. All remaining good jobs will be serviced in subsequent iterations in similar fashion ($\estGap$ per iteration), since the \currentprice never again exceeds $2^{\floor{\log(\estGap(\estGap+1))}} \leq \estGap (\estGap+1)$ (recall Lemma~\ref{l:upperGood}). 

\medskip
\noindent{\bf Results for \LinPow.} Figure~\ref{fig:linear-power} illustrates the results of our simulations when $\estGap=8$. As predicted by our upper bound on cost in Theorem~\ref{t:LinPow}, \LinPow achieves a lower cost than the adversary for $B$ sufficiently large. However, aligning with our analysis, 
we see that this cost advantage decreases as $M$ grows. In particular, between $M=1$ and $M=8$, we see that the cost ratio decreases from roughly $2\times10^3$ to $20$ for a fixed $B=9 \times 10^{11}$. Very similar results are observed for other values of $\estGap$, and so we omit these and note that, for this set of experiments, the algorithm's cost appears to be primarily sensitive to $M$.

%\begin{wrapfigure}[16]{r}{.55\textwidth}
%\centering
%\vspace{-20px}
\begin{figure}[t!]
\begin{center}
\includegraphics[width=0.55\textwidth]{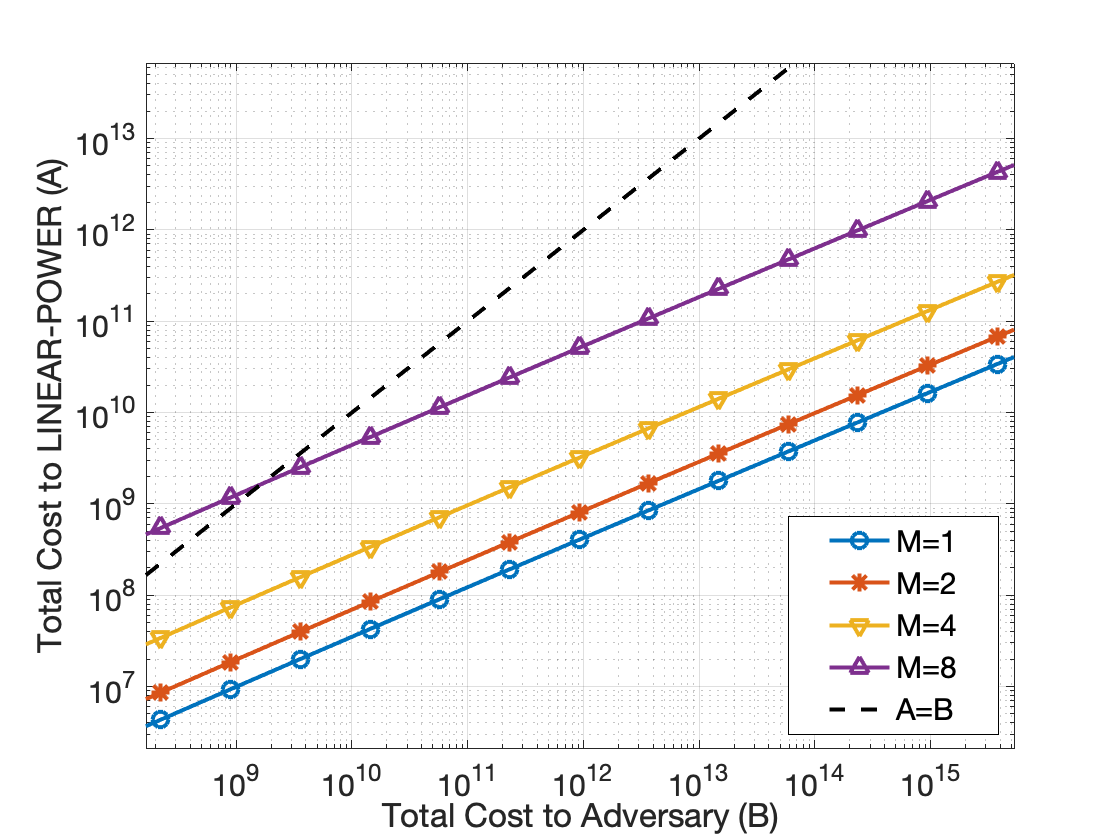}
\vspace{-0pt}
\caption{Our preliminary simulation results for \LinPow. Different values of $M$ are evaluated with $\estGap = 8$.}\label{fig:linear-power}
\end{center}
\end{figure}
%\vspace{-20px}
%\end{wrapfigure}

%%%%%%%%%%%%%%%%%%%%%%%%%%%%%%%%%%%%%%%%%%%%%%
%%%%%%%%%%%%%%%%%%%%%%%%%%%%%%%%%%%%%%%%%%%%%%
%%%%%%%%%%%%%%%%%%%%%%%%%%%%%%%%%%%%%%%%%%%%%%

%%%%%%%%%%%%%%%%%%%%%%%%%%%%%%%%%%%%%%%%%%%%%%
%%%%%%%%%%%%%%%%%%%%%%%%%%%%%%%%%%%%%%%%%%%%%%
%%%%%%%%%%%%%%%%%%%%%%%%%%%%%%%%%%%%%%%%%%%%%%

\section{Conclusion and Future Work}\label{s:conclusion-future-work}
 
We have provided two deterministic algorithms for DoS defense.  Our algorithms dynamically adjust the server's prices based on feedback from an estimator which estimates the number of good jobs in any previously observed time interval.  The accuracy of the estimator on the input, the number and temporal placement of the good and bad jobs, and the adversary's total spending are all unknown to the algorithm, but these all affect the algorithm's cost.  

Critically, during a significant attack the cost to our algorithm is {\it asymptotically} less than the cost to the adversary.  We believe this is a important property for deterring DoS attacks.   We have also given a lower bound for randomized algorithms, showing that our algorithmic costs are asymptotically tight, whenever $\estGap$ is a constant.  We note that, while our algorithms are deterministic, our lower bound holds even for randomized algorithms.   

Many interesting problems remain.  Are there less restrictive properties for the estimator which suffice to achieve a good pricing strategy?  Can we model both the clients and the server as selfish but rational agents?  In particular, can we design a mechanism that results in a Nash equilibrium, and also provides good cost performance as a function of the adversarial cost?\medskip

\noindent{\bf Acknowledgements.} This material is based upon work supported by the National Science Foundation under grant numbers CNS-2210299, CNS-2210300, and CCF-2144410.

%\vspace{-5pt}
%Also, are there other uses of a pricing scheme that can provide other information about the clients. 
 
%\bibliographystyle{plain} 
%\bibliography{DDoS-RB}

\end{document}